\documentclass[journal]{IEEEtran}
\usepackage{amsmath,amsfonts}
\usepackage{algorithmic}
\usepackage{array}
\usepackage{textcomp}
\usepackage{stfloats}
\usepackage{url}
\usepackage{verbatim}
\usepackage{graphicx}
\def\BibTeX{{\rm B\kern-.05em{\sc i\kern-.025em b}\kern-.08em
    T\kern-.1667em\lower.7ex\hbox{E}\kern-.125emX}}
\usepackage{balance}
\usepackage{verbatim}
\usepackage{amsfonts}
\usepackage{cleveref}
\usepackage{amssymb}
\usepackage{stfloats}
\usepackage{cite}
\usepackage{soul}
\usepackage{setspace}
\usepackage{graphicx}
\usepackage{psfrag}
\usepackage{amsmath}
\usepackage{array}
\usepackage{epstopdf}
\usepackage{algorithm}
\usepackage{graphicx} 
\usepackage{amsthm} 
\usepackage{lipsum}
\usepackage{verbatim} 
\usepackage{mathtools}
\usepackage{cuted}
\usepackage{bm}
\usepackage{booktabs}
\usepackage{subfigure}
\usepackage[font=small]{caption} 
\usepackage{xcolor}

\usepackage{ifthen}

\newtheorem{Lemma}{Lemma}
\newtheorem{Theorem}{Theorem}

\newtheorem{Proposition}{Proposition}
\newtheorem{Remark}{Remark}

\definecolor{revblue}{RGB}{0,70,140}
\definecolor{revred}{RGB}{170,0,0}

\begin{document}

\title{\fontsize{22.8pt}{28pt}\selectfont Multi-SPIN: Multi-Access Speculative Inference for Cooperative Token Generation at the Edge}
\author{ 
Haotian~Zheng,
Zhanwei~Wang,
Mingyao~Cui,
Chang~Cai,\\
Hongyang~Du,
and Kaibin~Huang,~\IEEEmembership{Fellow,~IEEE}%
  
   \thanks{H.~Zheng, Z.~Wang, M.~Cui, C.~Cai, H.~Du, and K. Huang are with the Department of Electrical and Computer Engineering, The University of Hong Kong, Hong Kong SAR, China (Email: \{htzheng, zhanweiw, mycui, changcai, duhy, huangkb\}@eee.hku.hk).
  Corresponding authors: K. Huang; Z. Wang.}\vspace{-3.5mm}
  }

\maketitle

\begin{abstract}
Speculative inference (SPIN) was originally developed as an efficient architecture to accelerate large language models (LLMs). 
In this work, we propose its distributed deployment to enable cooperative token generation in a multiuser edge system; its advantage is to effectively balance computational loads between resource-constrained devices and servers.
The resulting architecture, termed Multi-access SPIN (Multi-SPIN), utilizes on-device small language models (SLMs) to generate and upload candidate token drafts, alongside an edge server that operates the LLM to verify these drafts in parallel batches. 
Given the severe heterogeneity in users' computation and communication (C$^2$) capabilities, the draft length emerges as a critical control variable that influences node-level computation loads and multi-access latency, thereby governing the sum token goodput.
Consequently, considering frequency-division multiple access, we investigate the problem of so-called \emph{multi-access draft control}, referring to joint optimization of draft-length control and bandwidth allocation to maximize the sum token goodput.
We examine two cases: (1) homogeneous draft lengths across users to facilitate server-side batching, and (2) heterogeneous draft lengths to introduce a new dimension for goodput enhancement. 
By developing decomposition methods, we systematically reduce these complex optimizations into tractable sub-problems, which allow efficient draft control algorithms to be derived in closed form.
Our analysis shows that the optimal bandwidth allocation compensates users with weaker C$^2$ capabilities in the homogeneous case due to the batching synchronization requirements, whereas its heterogeneous-case counterpart rewards users with higher acceptance rates by relaxing such requirements.
Experiments utilizing Llama-2 and Qwen3.5 model pairs across tasks demonstrate that Multi-SPIN improves goodput by up to $88\%$ compared to heterogeneity-agnostic baselines.
\end{abstract}



\section{Introduction}

Deploying \emph{large language models} (LLMs) at the \emph{sixth-generation} (6G) network edge enables highly efficient,
low-latency generative AI services for mobile devices~\cite{6Groadmap,GX-CM-2020}.
However, resource-constrained edge servers often struggle with the heavy computational demands of LLMs
containing tens to hundreds of billions of parameters~\cite{liu2023resource,qu2025partialloading}.
\emph{Speculative inference} (SPIN) mitigates this by leveraging distributed on-device computing:
devices use lightweight \emph{small language models} (SLMs) to generate candidate token sequences, or \emph{drafts},
which are then uploaded to the server for parallel LLM verification~\cite{shao2025ai}.
As such single-pass verification is significantly less complex than sequential autoregressive generation,
SPIN drastically reduces the server's computational burden without sacrificing accuracy~\cite{leviathan2023fast}.
Nevertheless, scaling SPIN for multiuser systems introduces two key challenges.
The first is a local \emph{computation-and-communication} (C$^2$) bottleneck driven by on-device draft
generation and multi-access uplink transmissions.
The second is a global computation bottleneck resulting from the simultaneous verification of
multiple drafts at the server.
To overcome these limitations, we propose a \emph{multi-access SPIN} (Multi-SPIN) framework that utilizes draft batching for efficient parallel server-side verification.
As a key component of this framework, we present novel algorithms to jointly optimize draft lengths
and uplink bandwidth allocation, with the objective of maximizing the token goodput.

The advantages of SPIN can be clearer when compared with conventional edge LLMs.
In standard server-hosted LLM generation, each token requires a computationally expensive forward
pass and extensive memory access to retrieve the ``KV cache'', referring to historical intermediate
data essential for autoregressive generation~\cite{kwon2023efficient}.
Repeating this sequential process over a long token sequence and for multiple users severely strains
resource-constrained edge servers~\cite{zhang2024edgeshard}.
In contrast, the Multi-SPIN framework offloads much of this workload to many increasingly powerful
mobile processors by using local SLMs to generate draft token sequences~\cite{hao2024hybrid,xu2024edgellm}.
Although these drafts naturally deviate from the LLM's target distribution~\cite{miao2024specinfer,chen2023accelerating},
the edge server resolves this through a highly efficient verification process.
By computing the target token probabilities of multiple user drafts simultaneously in a single batched
forward pass, the server accepts valid tokens and calibrates any discrepancies to guarantee LLM-level
quality before returning the finalized tokens to the devices.

The current mainstream approach for alleviating computational load on edge servers is so-called
\emph{split inference}~\cite{li2019edge}.
It offloads early layers of a global model (i.e., deep convolutional neural network) to an edge device
while processing the remainder at the server.
The split point can be dynamically adjusted to balance computing loads and control wireless transmission overhead.
While effective for sensing and object recognition, this architecture incurs severe latency for autoregressive token generation in transformer-based LLMs as explained later.
Existing split inference research focuses on reducing communication overhead via split-point
optimization~\cite{shao2021communication, yan2022optimal} and feature compression and
transmission~\cite{wang2025revisiting,wang2026airbreath,zeng2024knowledge}.
Recent applications to edge LLMs have introduced practical techniques such as memory-aware layer
placement~\cite{zhang2024edgeshard}, over-the-air tensor parallelism~\cite{zhang2025llm-inference},
and activation-based routing of expert sub-models~\cite{xue2025wdmoe,wang2026spacemoe}.
However, these methods overlook the fundamental latency bottleneck of autoregressive generation.
In a split framework, a device cannot process a new token until the server generates and transmits the preceding one.
This sequential dependency requires numerous communication rounds to generate a full sequence;
each token still necessitates a computationally inefficient full forward pass through the global model.
In contrast, Multi-SPIN achieves higher communication efficiency by having the device upload an entire drafted sequence to the edge server in a single shot for LLM verification.
Furthermore, by verifying multiuser drafts in parallel within a single forward pass, server computation becomes inherently more scalable and efficient than the sequential processing mandated by split inference.

The advantages of SPIN have motivated recent studies on its edge deployment; however, the early work remains largely confined to simple point-to-point, single-user systems.
Existing research primarily focuses on mitigating uplink overhead via customized compression
strategies, such as importance-aware logit truncation~\cite{zheng2025communicationefficientcollaborativellminference}, draft-token selection~\cite{oh2024uncertainty}, and adaptive quantization~\cite{zhang2025quantize}.
Other studies attempt to boost the LLM token-acceptance rate by uploading multiple drafts, albeit at the expense of higher local C$^2$ overhead~\cite{zheng2025fastcollaborativeinferencedistributed}.
In view of prior work, the single-user approaches fail to address the complexities inherent to typical multiuser edge systems~\cite{wen2023task}.
The primary challenge lies in the severe heterogeneity across devices (i.e., varying C$^2$
capabilities and prompt characteristics), which naturally leads to diverse token-acceptance rates~\cite{chen2025spin}.
This heterogeneity necessitates the joint optimization of individual draft lengths to maximize the \emph{sum token goodput}, defined as the aggregate expected number of accepted tokens per second across all devices.
Such optimization must balance a fundamental tradeoff: while longer drafts benefit devices with high acceptance rates, they inherently prolong local computation and transmission latency, and vice versa.
Furthermore, although traditional radio resource management can mitigate multi-access latency, it yields sub-optimal \emph{end-to-end} (E2E) performance because it is decoupled from draft control, treating communication overhead independently of computational delay~\cite{lyu2024rethinking}.

The draft length is a unique control variable in Multi-SPIN that directly influences the sum token goodput.
Specifically, it determines the number of tokens each user generates and transmits, thereby driving the local computation loads, the uplink communication latency, and the server's verification overhead. 
For a local SLM with a given token-acceptance rate, an excessively long draft incurs unnecessary C$^2$ latency, as the server will reject many of the drafted tokens. 
On the other hand, an overly short draft diminishes server efficiency, as the fixed computational overhead of a batched forward pass is amortized over too few verified tokens. 
Because token-acceptance rates and wireless channel conditions vary across devices, the optimal control of multiuser draft lengths becomes important for maximizing the sum token (generation) goodput.
Addressing the issue for a multiuser system gives rise to a novel problem, termed \emph{multi-access draft control}, referring to joint control of draft lengths and bandwidth allocation to maximize the sum token goodput.
Formulating and solving this problem forms the core contribution of this work.

The work is arguably the first study of distributed SPIN deployment in a multiuser wireless-edge system by proposing the Multi-SPIN framework. 
As summarized below, it comprises a set of algorithms to solve the optimal multi-access draft control problem under different system settings.
The key contributions and findings of this work are summarized as follows.

\begin{itemize}

    \item \textbf{Multi-access Draft Control:}
    To facilitate batched processing at the server, we assume uniform draft lengths across all users under centralized server control. The sum token goodput maximization problem is then optimally decoupled into two sub-problems.
    The first sub-problem focuses on optimal draft-length control under a given bandwidth allocation. Here, we identify a fundamental content--latency tradeoff: as the draft length increases, the number of accepted tokens gradually saturates, whereas the E2E latency grows linearly. This renders the sum token goodput a unimodal function of the draft length, allowing the unique optimum to be derived in closed form. The optimal draft length is shown to increase monotonically with both the acceptance rate and the verification latency, thereby amortizing the heavy verification overhead over more accepted tokens.
    The second sub-problem addresses bandwidth allocation to minimize the multi-access C$^2$ latency, where the derived optimal strategy allocates larger bandwidth to devices with weaker C$^2$ capabilities due to the synchronization requirement inherent to server-side batching.

    \item \textbf{Multi-access Draft Control with Heterogeneous Lengths:}
    Next, we permit heterogeneous draft lengths across users, which opens a new dimension for goodput enhancement. To maintain compatibility with server-side batched processing, draft-length variations are accommodated via zero-padding. Although this flexibility deepens the coupling between draft control and bandwidth allocation, the problem can still be optimally decomposed for a low-complexity solution. This is achieved by using the multi-access latency as an intermediate variable, which yields two nested sub-problems: draft-length control for a given multi-access latency, and bandwidth allocation to minimize this latency. Solving both in closed form reduces the original high-dimensional optimization to a highly efficient two-dimensional search. The resulting optimal strategy assigns longer drafts and greater bandwidth to devices with higher acceptance rates---in contrast to the bandwidth allocation derived for the homogeneous case.

    \item \textbf{Experimental Results:} 
    Experiments on Llama-2 and Qwen3.5 model pairs confirm the derived tradeoff and the substantial goodput gain of Multi-SPIN over heterogeneity-agnostic baselines.
    The heterogeneous-length scheme consistently surpasses its homogeneous counterpart, with the gap enlarging as the network scales to more devices.
\end{itemize}

The remainder of this paper is organized as follows. Sec.~\ref{sec:model_metric} presents the system model and performance metrics. 
Sec.~\ref{sec:protocol_overview} introduces the Multi-SPIN protocol and problem formulation. 
Sec.~\ref{sec:draft_length_teba} studies the multi-access draft control problem under uniform draft lengths.
Sec.~\ref{sec:joint_scheduling} extends the analysis to the general regime with heterogeneous draft lengths. 
Sec.~\ref{sec:experiment_results} provides the experimental results, and Sec.~\ref{sec:conclusions} concludes the paper.

\section{Models and Metrics}
\label{sec:model_metric}

The proposed Multi-SPIN framework considers a single-cell edge network in which an edge server, which hosts an LLM, provides generative AI services to $K$ distributed devices.
Each device executes a local SLM to generate and upload a draft (i.e., token sequence) to the server, where the server-side LLM collects and verifies the multiuser drafts in a batch.
The associated models and performance metrics are described in the subsections.

\subsection{Speculative-inference Operations and Models}
Each epoch of the token generation process in SPIN spans multiple rounds, each comprising the following two sequential operations: \emph{local drafting} and \emph{server-side verification}.

\subsubsection{Local Drafting}
\label{sec:local_draft_model}
Consider an arbitrary SPIN round for device $k\in\{1,\dots,K\}$. In the first step, the device generates a token draft using its on-device SLM, parameterized by $\Phi^{\mathsf{S}}$.
This step follows the classical autoregressive inference paradigm, where the draft is generated token by token, as illustrated in the lower part of Fig.~\ref{fig:inference_model_demo}. 
The associated notation is introduced as follows.
Let $M_k$ and $L_k$ denote the numbers of tokens in the prefix sequence and the draft of device $k$, respectively. Let $\hat{X}_{k}\triangleq(\hat{x}_{k,1},\ldots,\hat{x}_{k,L_k})$ denote the draft generated by device $k$ conditioned on the prefix sequence $X_{k}=(X_k^{\sf pt},x_{k,1},\dots,x_{k,M_k})$, where $X_k^{\sf pt}$ denotes the token sequence corresponding to the input prompt of device $k$. Each token takes a value from the vocabulary $\mathcal{V}$, whose cardinality is $V=|\mathcal{V}|$.

For the $\ell$-th token, denoted by $\hat{x}_{k,\ell}$, the SLM computes, based on the current prefix sequence, a probability vector that characterizes the conditional \emph{probability mass function} (PMF) of $\hat{x}_{k,\ell}$ over the vocabulary:
\begin{equation}
\label{eq:prob_vector}
\begin{split}
\mathbf{p}^{\mathsf{S}}_{k,\ell}
&\triangleq
\left[
P^{\mathsf{S}}_{k,\ell}(1),\dots,P^{\mathsf{S}}_{k,\ell}(v),\dots,P^{\mathsf{S}}_{k,\ell}(V)
\right]^{\sf T}\\
&=f_{\Phi^{\mathsf{S}}}\!\left( X_{k}, \hat{x}_{k,1},\ldots,\hat{x}_{k,\ell-1}\right).
\end{split}
\end{equation}
Here, $P^{\mathsf{S}}_{k,\ell}(v)=\Pr(\hat{x}_{k,\ell}=v \mid X_k,\hat{x}_{k,1},\ldots,\hat{x}_{k,\ell-1})$ denotes the probability assigned by the SLM to vocabulary token $v\in\mathcal{V}$, and $f_{\Phi^{\mathsf{S}}}$ denotes the inference mapping of the SLM. 
Each token realization is obtained by sampling $\hat{x}_{k,\ell} \sim \mathbf{p}^{\mathsf{S}}_{k,\ell}$, and the draft $\hat{X}_{k}$ is formed by recursively computing in \eqref{eq:prob_vector} and sampling across positions $\ell=1,\ldots,L_k$.
Since this sequential process requires one SLM forward pass per token, the total local drafting latency is given by
\begin{equation}
T_k^{\mathsf{dr}} = L_k T_k^{\mathsf{S}},
\label{eq:draft_latency}
\end{equation}
where $T_k^{\mathsf{S}}$ is the per-token inference latency of device $k$.
The generated draft and associated probability vectors are then forwarded to the edge server for verification.

\begin{figure}[t]
    \centering
\includegraphics[width=0.95\linewidth]{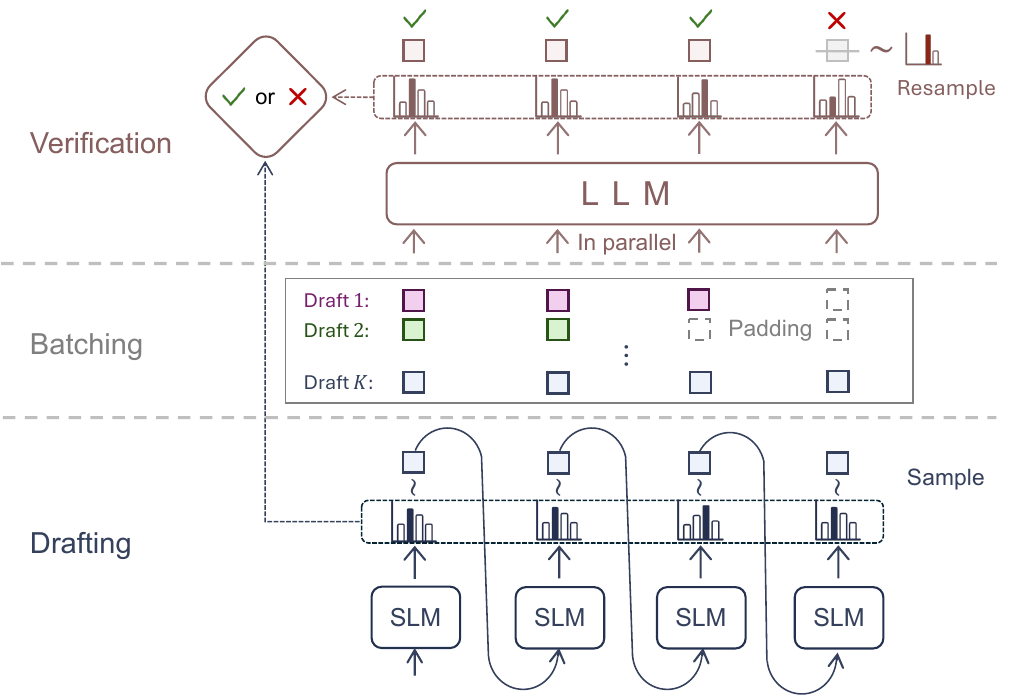}
    \caption{Illustration of speculative-inference operations and models.}
    \label{fig:inference_model_demo}
\end{figure}

\subsubsection{Server-side Verification}
\label{sec:server_verify}

The edge server uses the LLM, parameterized by $\Phi^{\mathsf{L}}$, to verify the draft $\hat{X}_{k}$ generated by device $k$, as illustrated in the upper part of Fig.~\ref{fig:inference_model_demo}. 
Unlike autoregressive decoding, which requires one sequential forward pass per token, the server-side LLM verifies the entire draft in \emph{a single parallel forward pass}.
Specifically, for each drafted position $\ell$, the LLM computes a probability vector that characterizes the conditional PMF over the vocabulary, given by
\begin{equation}
\label{eq:prob_vector_LLM}
\begin{split}
\mathbf{p}^{\mathsf{L}}_{k,\ell}
&\triangleq
\left[
P^{\mathsf{L}}_{k,\ell}(1),\dots,P^{\mathsf{L}}_{k,\ell}(v),\dots,P^{\mathsf{L}}_{k,\ell}(V)
\right]^{\sf T}\\
&=f_{\Phi^{\mathsf{L}}}\!\left( X_{k}, \hat{x}_{k,1},\ldots,\hat{x}_{k,\ell-1}\right),
\end{split}
\end{equation}
where $P^{\mathsf{L}}_{k,\ell}(v)=\Pr(\hat{x}_{k,\ell}=v\mid X_k,\hat{x}_{k,1},\ldots,\hat{x}_{k,\ell-1})$ denotes the probability assigned by the LLM to token $v\in\mathcal{V}$, and $f_{\Phi^{\mathsf{L}}}$ denotes the inference mapping of the LLM. 

Based on the probability vectors of the SLM and LLM, the server performs token verification as follows. 
For the draft token $\hat{x}_{k,\ell}$, the verification outcome is modeled as a Bernoulli random variable $A_{k,\ell} \sim \mathrm{Bernoulli}(\beta_{k,\ell})$,
where $A_{k,\ell}=1$ and $A_{k,\ell}=0$ indicate acceptance and rejection, respectively, and $\beta_{k,\ell}$ denotes the acceptance probability, specified as
\begin{equation}
\beta_{k,\ell}
\triangleq
\min\!\left\{1,\,
\frac{P^{\mathsf{L}}_{k,\ell}\!\left(\hat{x}_{k,\ell}\right)}
{P^{\mathsf{S}}_{k,\ell}\!\left(\hat{x}_{k,\ell}\right)}
\right\}.
\label{eq:accept_prob}
\end{equation}
Here, $P^{\mathsf{S}}_{k,\ell}\!\left(\hat{x}_{k,\ell}\right)$ and $P^{\mathsf{L}}_{k,\ell}\!\left(\hat{x}_{k,\ell}\right)$ are the probabilities assigned in~\eqref{eq:prob_vector} and~\eqref{eq:prob_vector_LLM}, respectively, to the drafted token $\hat{x}_{k,\ell}$.

The server examines the outcomes sequentially and identifies the first rejected position $\ell_k^{\sf rej}=\min\{\ell : A_{k,\ell}=0,\, \ell=1,\dots,L_k\}$.
All tokens with $\ell<\ell_k^{\sf rej}$ are accepted, and the rejected position is replaced by a calibrated token $x^{\sf cal}_{k,\ell_k^{\sf rej}}$ sampled from a calibrated distribution~\cite{leviathan2023fast}.
The resulting output token sequence after LLM verification, denoted by $\tilde{X}_k$, is provided by
\begin{equation}
\tilde{X}_k=
\begin{cases}
(\hat{x}_{k,1},\dots,\hat{x}_{k,\ell_k^{\sf rej}-1},x^{\sf cal}_{k,\ell_k^{\sf rej}}), & \ell_k^{\sf rej}\leq L_k,\\
(\hat{x}_{k,1},\dots,\hat{x}_{k,L_k},\tilde{x}_{k,L_k+1}), & \text{otherwise},
\end{cases}
\end{equation}
where $\tilde{x}_{k,L_k+1}\sim \mathbf{p}^{\mathsf{L}}_{k,L_k+1}$ is one additional token sampled from the LLM distribution when all $L_k$ drafted tokens are accepted.
Appending $\tilde{X}_k$ to the current prefix sequence yields the prefix for the next round:
\begin{equation}
\label{eq:prefix_update}
X_{k}\leftarrow (X_k, \tilde{X}_k).
\end{equation}

On one hand, tokens generated by the above drafting-verification cooperation match the accuracy of direct LLM inference~\cite{chen2023accelerating}.
On the other hand, SPIN enables the LLM to evaluate all drafted tokens in parallel, approximately reducing the generation complexity from $\mathcal{O}(L\cdot \left|\Phi^{\mathsf{L}}\right|)$ to $\mathcal{O}(\left|\Phi^{\mathsf{L}}\right|)$~\cite{pope2023efficiently, nvidia_a100_datasheet}.

To reduce the verification overhead of serving $K$ devices, we consider adopting \emph{batching} at the server, packing the received drafts into a single batch tensor and processing them in one forward pass of the LLM~\cite{qian2024bass}.
Note that in the case of heterogeneous draft lengths, zero-padding is needed to enable batching, where shorter drafts are padded to a uniform length before being stacked into the batch tensor.
Based on empirical observations on batched LLM inference~\cite{su2023synergy}, the batched verification latency is modeled as
\begin{equation}
\label{eq:batched_validation}
T^{\mathsf{ver}}(K)=T^{\mathsf{fix}} + KT^{\mathsf{lin}},
\end{equation}
where $T^{\mathsf{fix}}$ captures fixed overhead such as GPU kernel launches and $T^{\mathsf{lin}}$ is the incremental latency per additional draft in the batch\footnote{To keep the analysis focused, we consider $T^{\mathsf{ver}}$ as approximately independent of the draft length. This draft-length-agnostic modeling of verification latency is widely adopted in the literature, e.g.,~\cite{leviathan2023fast, pope2023efficiently}.}.

\subsection{Multiple Access Model}
\label{subsubsec:token_tx}

In the proposed Multi-SPIN system, orthogonal multiple access is adopted for uploading drafts from $K$ distributed devices to the edge server. Specifically, the system employs \emph{orthogonal frequency-division multiple access} (OFDMA) over a broadband uplink channel. For simplicity, the number of subcarriers is assumed sufficiently large such that bandwidth allocation can be approximated as continuous~\cite{wu2025resource}.
For device $k$, let $p_k$ and $B_k$ denote the transmit power 
and allocated bandwidth, respectively.
Let $H_k = |h_k|^2$ denote the channel power gain, where 
$h_k \sim \mathcal{CN}(0,\,\bar{H}_k)$ with $\bar{H}_k$ 
representing the average channel power gain that incorporates 
path loss and large-scale fading.
The channel is assumed to remain constant within one round 
and to be known at the server~\cite{you2016energy}.
The resulting uplink rate is given by
\begin{equation}
  R_k = B_k \log_2\!\left(1 + \frac{p_k\,H_k}{N_0\,B_k}\right)
      \triangleq B_k\,r_k,
  \label{eq:uplink_rate}
\end{equation}
where $N_0$ denotes the noise power spectral density and $r_k$ is defined to be the uplink spectrum efficiency.

Consider uploading the token draft of device $k$ with draft length $L_k$. 
To reduce the uplink overhead, each device uploads, for every drafted token, only the $|\hat{\mathcal V}|$ retained probability values and their corresponding vocabulary indices, rather than the full-dimensional probability vector.
Then, each probability value is quantized using $Q_{\sf B}$ bits, where $Q_{\sf B}$ is chosen sufficiently large such that the resulting quantization error is negligible (e.g., $Q_{\sf B}=16$). 
Since each vocabulary index requires $\left\lceil \log_2 V \right\rceil$ bits, the uplink transmission latency is given by
\begin{equation}
T_k^{\sf tx}
=
\frac{Q_{\mathsf{tok}}}{B_k r_k}L_k,
\label{eq:logit_latency}
\end{equation}
where $Q_{\mathsf{tok}}\triangleq
|\hat{\mathcal V}|\left(Q_{\sf B}+\left\lceil \log_2 V \right\rceil\right)$ is the per-token communication overhead and $V$ is the vocabulary cardinality. 
The downlink transmission latency is neglected due to the ample transmit power of the edge server.

\subsection{Definition of Sum Token Goodput}

To evaluate the performance of the multi-access SPIN system, we consider the \emph{sum token goodput}, defined as the aggregate expected number of accepted tokens generated per second. The exact mathematical definition is given below.

\subsubsection{Expected Number of Accepted Tokens}
Consider an arbitrary SPIN round for device $k$, where the SLM generates a token draft $\hat{X}_{k}$ of length $L_k$. Owing to the randomness of server-side verification, the number of accepted tokens, denoted by $N_k$, is a random variable. 
To enable tractable analysis, the token-acceptance events are widely considered as i.i.d. across drafted positions~\cite{leviathan2023fast}. Accordingly, the acceptance probability can be obtained as the mean verification outcome
\begin{equation}
\alpha_k \triangleq \mathbb{E}_{X_k^{\sf pt},\hat{x}_{k,\ell}}\!\left[A_{k,\ell}\right]
= \mathbb{E}_{X_k^{\sf pt},\hat{x}_{k,\ell}}\!\left[\beta_{k,\ell}\right],
\label{eq:alpha_def}
\end{equation}
where the expectation is taken over prompts $X_k^{\sf pt}\in\mathcal{D}_k$ and the drafted tokens generated by the SLM. 
The quantity $\alpha_k$, termed the \emph{acceptance rate}, can be statistically estimated for a given SLM--LLM pair and prompt categories~\cite{parallelSD, zhang2025quantize}.

Under this approximation, the PMF of $N_k$ is given by
\begin{equation}
\Pr(N_k=\ell)=
\begin{cases}
\alpha_k^{\ell-1}(1-\alpha_k), & \ell=1,2,\dots,L_k,\\
\alpha_k^{L_k}, & \ell=L_k+1.
\end{cases}
\end{equation}
Accordingly, for a fixed draft length $L_k$ of device $k$, the expected number of accepted tokens becomes~\cite{parallelSD}
\begin{equation}
\mathbb{E}[N_k\mid L_k]
=
\sum_{\ell=0}^{L_k}\alpha_k^{\ell}
=
\frac{1-\alpha_k^{L_k+1}}{1-\alpha_k}.
\label{eq:expected_tokens}
\end{equation}

\subsubsection{Sum Token Goodput}
\label{sec:token_goodput}
In Multi-SPIN, the sum token goodput is defined as the expected number of accepted tokens aggregated over all $K$ devices, divided by the E2E latency of one Multi-SPIN round.
Denote ${T^{\sf e2e}(\mathcal{B},\mathcal{L})}$ as the E2E execution latency of completing one SPIN round for all $K$ devices, which will be quantified in the next section. Then, the sum goodput of Multi-SPIN is given by

\begin{equation}
\label{eq:e2e_metric}
  \tau(\mathcal{B},\mathcal{L}) = \frac{ \sum_{k=1}^K \mathbb{E}[N_k\mid L_k] }{T^{\sf e2e}(\mathcal{B},\mathcal{L})},
\end{equation}
where $\mathcal{B}=\{B_k\}_{k=1}^K$ and $\mathcal{L}=\{L_k\}_{k=1}^K$ denote the sets of bandwidth allocations and draft lengths of the devices.

\begin{figure*}[!t]
    \centering

    \subfigure[Multi-SPIN system.]{%
        \includegraphics[width=0.57\textwidth]{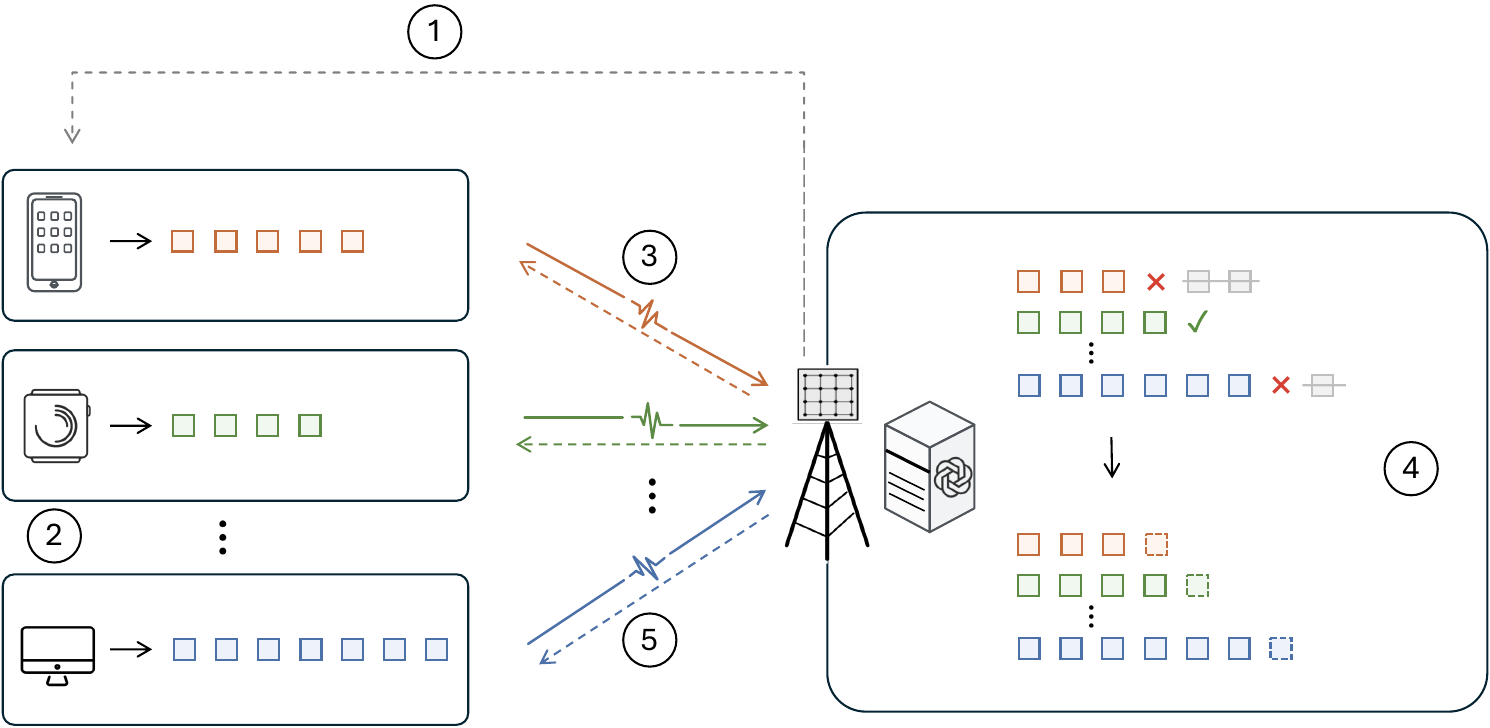}%
        \label{fig:overview_a}%
    }
    \hfill
    \subfigure[Multi-SPIN operations and protocol.]{%
        \includegraphics[width=0.38\textwidth]{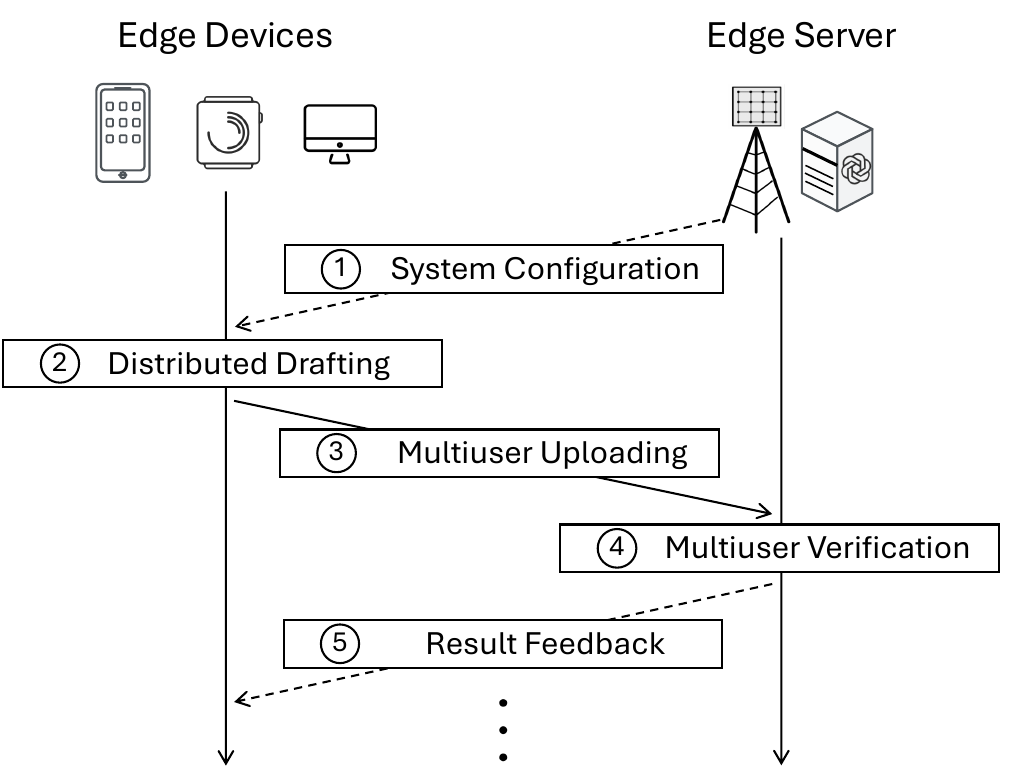}%
        \label{fig:overview_b}%
    }

    \caption{Overview of the Multi-SPIN framework.}
    \label{fig:overview}
\end{figure*}

\section{Protocol and Problem Formulation}
\label{sec:protocol_overview}
In this section, we present the Multi-SPIN protocol for coordinating parallel SPIN tasks among distributed devices. 
Then, the problem of optimal multi-access draft control is formulated.

\subsection{Multi-SPIN Protocol}

The steps of the Multi-SPIN protocol are illustrated in Fig.~\ref{fig:overview} and elaborated as follows.

\begin{enumerate}
\item \emph{System Configuration:} 
At the beginning of each round, each device reports its task profile (including the acceptance rate) and computation speed to the server.
The server then measures the uplink channel conditions and determines the draft lengths and bandwidth allocations by solving the multi-access draft control problem in Sec.~\ref{sec:goodput_problem}.
The resulting configurations are delivered to all devices.

\item \emph{Distributed Drafting:}
Each device independently generates a draft of $L_k$ tokens using its local SLM, incurring a local computation latency in~\eqref{eq:draft_latency}.

\item \emph{Multiuser Draft Uploading:} 
Each device transmits the predicted PMF vectors in \eqref{eq:prob_vector} together with the drafted token indices to the edge server via its assigned OFDMA channel, resulting in an uplink transmission latency in \eqref{eq:logit_latency}.

\item \emph{Multiuser Draft Verification:}
Upon receiving all $K$ drafts, the edge server performs batched verification via a single LLM forward pass.
For each draft, tokens are accepted or rejected via~\eqref{eq:accept_prob}, and the rejected one is replaced with a calibrated token.

\item \emph{Result Feedback:}
The server returns the verified token sequence to each device, which updates its prefix via~\eqref{eq:prefix_update}. 
If the current request is completed, the device proceeds to its next request.
The system then loops back to Step~1 for the next round.

\end{enumerate}

\subsection{Goodput Maximization Problem}
\label{sec:goodput_problem}

This subsection first characterizes the explicit expression of the Multi-SPIN goodput in~\eqref{eq:e2e_metric} and then formulates the corresponding optimization problem to maximize it.
We consider a uniform-length regime where all devices adopt a uniform draft length, i.e., $L_k=L$ for all $k$, and a common acceptance rate, i.e., $\alpha_k=\alpha$ for all $k$. 
This setting is practically reasonable when devices have comparable C$^2$ capabilities and serve prompts from the same task family, as it reduces control complexity and enables synchronized batch verification at the server~\cite{leviathan2023fast}.
This problem is extended into heterogeneous scenarios for device-specific acceptance rate and draft-length control in Sec.~\ref{sec:joint_scheduling}.

Consider an arbitrary round.
Under the uniform-length regime, the expected number of accepted tokens in~\eqref{eq:expected_tokens} becomes identical across devices, so the aggregate over all $K$ devices is given by
\begin{equation}\label{eq:ENk_uniform}
  \sum_{k=1}^K\mathbb{E}[N_k \mid L_k=L] = K\cdot\frac{1-\alpha^{L+1}}{1-\alpha}.
\end{equation}
Local drafting and uploading are performed in parallel across devices. 
Combining \eqref{eq:draft_latency} and \eqref{eq:logit_latency}, we define the \emph{multi-access latency} $T^{\mathsf{ma}}(\mathcal{B},L)$ as the time required for all $K$ devices to generate and upload their drafts over the shared uplink, which is dominated by the slowest device and given by
\begin{equation}\label{eq:Tde}
  T^{\mathsf{ma}}(\mathcal{B},L)
  = L\,\max_{k}\!\left\{T_k^{\mathsf{S}}
    + \frac{Q_{\mathsf{tok}}}{B_k\, r_k}\right\}.
\end{equation}
The maximum arises because the server must wait for the slowest device to complete draft generation and token uploading.
As the multi-access and server-side phases execute sequentially within each round, the E2E round latency is
\begin{equation}
\label{eq:e2e_exe_lat}
    T^{\mathsf{e2e}}(\mathcal{B},L) = T^{\mathsf{ma}}(\mathcal{B},L) + T^{\mathsf{ver}},
\end{equation}
where $T^{\mathsf{ver}}$ denotes the batched verification latency in \eqref{eq:batched_validation}.
Substituting~\eqref{eq:ENk_uniform} and~\eqref{eq:e2e_exe_lat} into \eqref{eq:e2e_metric} yields the sum goodput of the Multi-SPIN system, given by
\begin{equation}\label{eq:goodput_uniform}
  \tau(\mathcal{B},L)
  = \frac{K\!\left(1-\alpha^{L+1}\right)}
         {\left[L\,\hat{T}^{\mathsf{ma}}(\mathcal{B}) + T^{\mathsf{ver}}\right](1-\alpha)},
\end{equation}
where $\hat{T}^{\mathsf{ma}}(\mathcal{B}) \triangleq \max_{k}\!\bigl\{T_k^{\mathsf{S}} + Q_{\mathsf{tok}}/(B_k r_k)\bigr\}$ denotes the per-token multi-access latency.
Maximizing this goodput requires jointly controlling the draft length and the bandwidth allocation, giving rise to the multi-access draft control problem, formulated as
\begin{align*}
  (\mathrm{P1})\quad
  \max_{\mathcal{B}, \, L} &\quad \tau(\mathcal{B},L) \notag\\
  \mathrm{s.t.} &\quad L \in \mathbb{Z}_{+}, \notag\\
  &\quad B_k > 0,\;\forall\, k, \notag\\
  &\quad \textstyle\sum_{k} B_k \le B,
\end{align*}
where $\mathcal{B}=\{B_k\}_{k=1}^{K}$ collects the per-device bandwidth allocations and $B$ denotes the total system bandwidth budget. 
In problem~(P1), the integer variable $L$ and the continuous variables $\{B_k\}$ are coupled through the straggler-limited latency in~\eqref{eq:Tde}. 
The next section exploits this structure to decouple bandwidth allocation from draft-length control and derive a closed-form optimal solution.

\section{Multi-access Draft Control}
\label{sec:draft_length_teba}

In this section, we assume uniform draft lengths across all users to facilitate batched processing at the server. Under this assumption, we design an optimal multi-access draft control algorithm by solving the goodput maximization problem~(P1). To design an efficient algorithm, we adopt a decomposition approach that optimally separates the original problem into two distinct sub-problems, both of which are subsequently solved in closed form.

\subsection{Optimal Problem Decomposition}
\label{subsec:problem_decoupling}

This subsection decouples problem~(P1) by exploiting the structure of the goodput expression in~\eqref{eq:goodput_uniform}. 
A key observation is that the per-token multi-access latency $\hat{T}^{\mathsf{ma}}(\mathcal{B})$ is the only term through which $\mathcal{B}$ affects the goodput. Specifically, the numerator of~\eqref{eq:goodput_uniform} depends only on $L$, whereas the denominator depends on $\mathcal{B}$ only through $\hat{T}^{\mathsf{ma}}(\mathcal{B})$. 
For any fixed $L>0$, $\tau(\mathcal{B},L)$ is monotonically decreasing in $\hat{T}^{\mathsf{ma}}$. 
Therefore, maximizing the goodput with respect to $\mathcal{B}$ is equivalent to minimizing $\hat{T}^{\mathsf{ma}}$. 
This observation leads to the following decomposition of problem~(P1).

\subsubsection{Bandwidth Allocation}
The bandwidth allocation problem seeks to minimize the per-token multi-access latency.
Since the constant factor $L$ in the multi-access latency $T^{\mathsf{ma}}(\mathcal{B},L)=L\,\hat{T}^{\mathsf{ma}}(\mathcal{B})$ does not affect the minimizer over~$\mathcal{B}$, the sub-problem is formulated as a min-max optimization problem, given by
\begin{equation*}
\begin{aligned}
(\mathrm{P1.1})\quad
\min_{\mathcal{B}}\quad & \hat{T}^{\mathsf{ma}} (\mathcal{B})\\
\mathrm{s.t.}\quad
& \sum_{k} B_k \le B,\\
& B_k > 0, \;\forall\, k.
\end{aligned}
\label{eq:inner_simplified}
\end{equation*}
Let $\mathcal{B}^\star$ denote the optimal solution of~(P1.1) and $\vartheta^\star \triangleq \hat{T}^{\mathsf{ma}} (\mathcal{B}^\star)$ the corresponding minimum per-token multi-access latency.

\subsubsection{Draft-Length Control}
With $\vartheta^\star$ determined, all dependence on $\mathcal{B}$ is eliminated, and the goodput in~\eqref{eq:goodput_uniform} reduces to a univariate function of $L$:
\begin{equation}
\tau(L)=\frac{K\!\left(1-\alpha^{L+1}\right)}{(L\,\vartheta^\star+T^{\mathsf{ver}})(1-\alpha)}.
\label{eq:throughput_obj}
\end{equation}
The multi-access draft control problem reduces to the draft-length control problem, given by
\begin{equation*}
\begin{aligned}
(\mathrm{P1.2})\quad
\max_{L \in \mathbb{Z}_+}\quad &
\tau(L).
\end{aligned}
\label{eq:uniform_opt_problem}
\end{equation*}

The two decoupled Problems~(P1.1) and~(P1.2) are solved in closed form in Sections~\ref{subsec:teba_commonL_calibrated} and~\ref{subsec:sec4_opt_commonL}, respectively.

\subsection{Optimal Bandwidth Allocation}
\label{subsec:teba_commonL_calibrated}

This subsection solves the bandwidth allocation problem~(P1.1) obtained from the decoupling in Sec.~\ref{subsec:problem_decoupling}.
Problem~(P1.1) minimizes the maximum of $K$ terms $T_k^{\mathsf{S}}+Q_{\mathsf{tok}}/(B_k r_k)$, each depending on a single bandwidth variable $B_k$.
Owing to this separable structure, the unique optimum is attained when all $K$ terms are equalized~\cite{yang2026optimal}, as stated in Lemma~\ref{prop:teba_solution}.

\begin{Lemma}[Optimal Bandwidth Allocation]
\label{prop:teba_solution}
Problem~(P1.1) admits a unique optimal bandwidth allocation, given by 
\begin{equation}
B_k^\star=\frac{Q_{\mathsf{tok}}}{r_k\!\left(\vartheta^\star-T_k^{\mathsf{S}}\right)},
\label{eq:bk_star}
\end{equation}
where $\vartheta^\star$ is determined as the unique root of
\begin{equation}
\sum_{k=1}^{K}\frac{Q_{\mathsf{tok}}}{r_k\!\left(\vartheta^\star-T_k^{\mathsf{S}}\right)}=B,
\qquad \vartheta^\star>\max_k \ T_k^{\mathsf{S}}.
\label{eq:vartheta_star_root}
\end{equation}
\end{Lemma}

The allocation~\eqref{eq:bk_star} satisfies
$T_k^{\mathsf{S}}+Q_{\mathsf{tok}}/(B_k^\star r_k)=\vartheta^\star$ for all $k$, meaning that the optimal bandwidth allocation equalizes the per-token multi-access latency across all devices, regardless of their heterogeneous computation times $\{T_k^{\mathsf S}\}$ and channel conditions $\{r_k\}$.
Moreover, $\vartheta^\star$ decreases monotonically with the total bandwidth budget $B$, as verified by
$\frac{d\vartheta^\star}{dB}
=
\left(-\sum_{k=1}^K \frac{Q_{\mathsf{tok}}}{r_k(\vartheta^\star-T_k^{\mathsf S})^2}\right)^{-1}
< 0$.
It confirms that a larger bandwidth budget reduces the communication latency.

\subsection{Optimal Draft Length}
\label{subsec:sec4_opt_commonL}

Building on Lemma~\ref{prop:teba_solution}, this subsection determines the optimal draft length that maximizes the sum goodput.
We solve the draft-length control problem~(P1.2) by first establishing the unimodality of the goodput and then deriving a closed-form optimum via continuous relaxation.

Since problem~(P1.2) is an integer program, we relax $L\in\mathbb{Z}_+$ to a continuous variable $\tilde{L}>0$ to enable closed-form analysis.
The relaxed goodput function is given by
\begin{equation}
\tilde{\tau}(\tilde{L})=\frac{K\left(1-\alpha^{\tilde{L}+1}\right)}{(\tilde{L}\,\vartheta^\star+T^{\mathsf{ver}})(1-\alpha)}, \quad \tilde{L}>0.
\label{eq:continuous_obj}
\end{equation}
Then, the maximization of $\tilde{\tau}(\tilde{L})$ in~\eqref{eq:continuous_obj} is formulated as
\begin{equation*}
\begin{aligned}
(\mathrm{P1.3}) \quad \max_{\tilde{L} > 0} \;\; & \tilde{\tau}(\tilde{L}).
\end{aligned}
\label{eq:uniform_opt_problem_approx}
\end{equation*}

\begin{figure}[t!]
    \centering
    \subfigure[Llama-2 Pair]{
        \includegraphics[width=0.45\columnwidth]{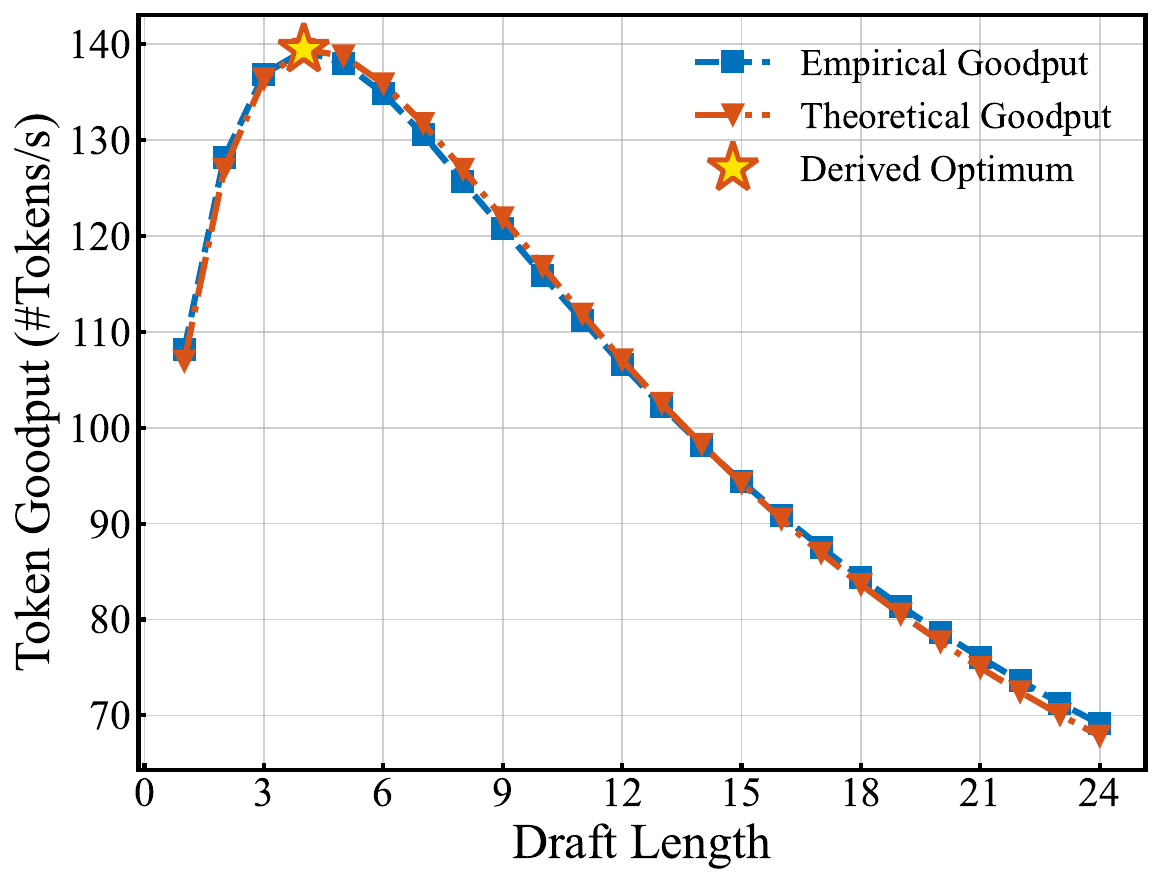}
    }
    \hfill
    \subfigure[Qwen3.5 Pair]{
        \includegraphics[width=0.45\columnwidth]{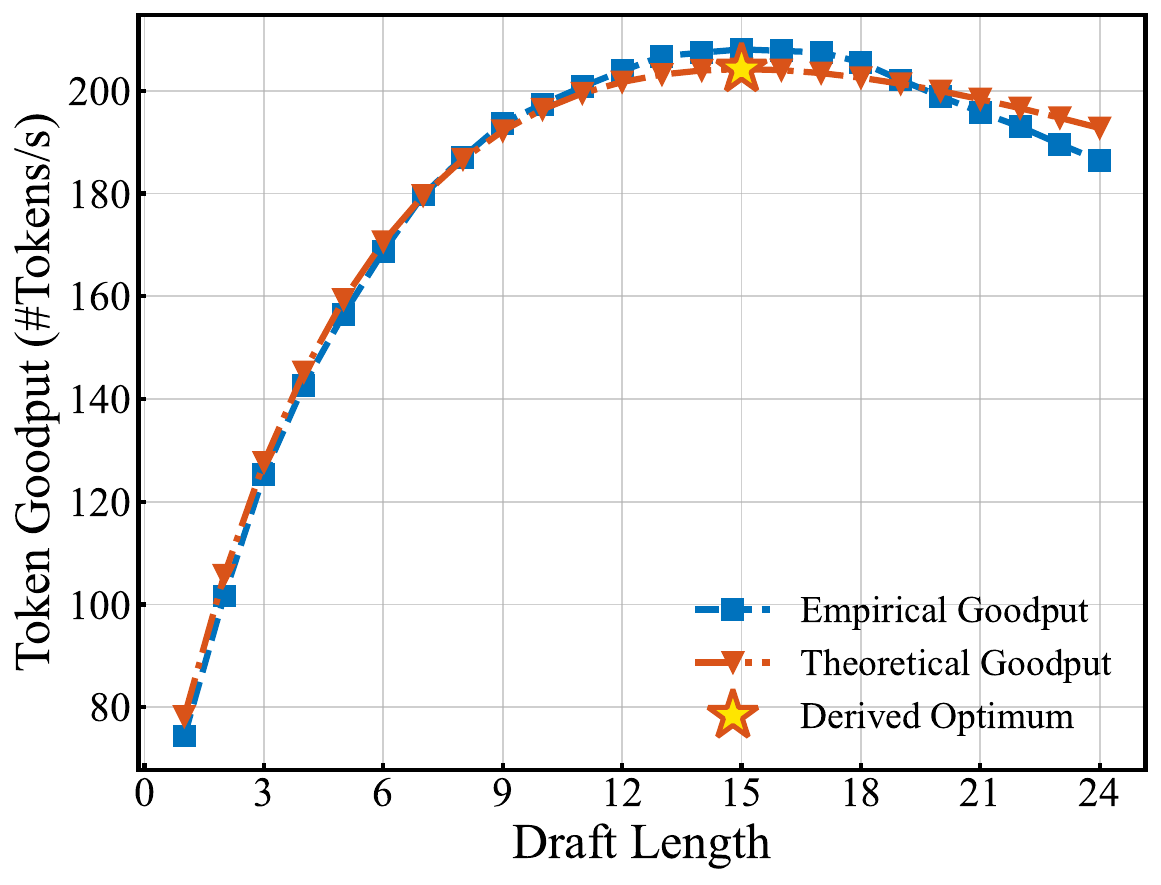}
    }
    \caption{Empirical and theoretical sum token goodput versus draft length for Llama-2 and Qwen3.5 models on the GSM8K dataset.}
    \label{fig:trade_l}
\end{figure}

As illustrated in Fig.~\ref{fig:trade_l}, the sum token goodput exhibits a \emph{content--latency tradeoff} with respect to the draft length $L$.
The close match between the empirical and theoretical curves validates the accuracy of the derived goodput model.
Increasing $L$ admits more verified tokens per round, though the marginal gain progressively shrinks, as a drafted token is accepted only if all preceding ones are.
Meanwhile, each additional token adds one unit of per-token multi-access latency $\vartheta^\star$, so the E2E latency grows linearly in $L$.
The goodput numerator thus saturates while its denominator grows linearly, yielding a unimodal function with a unique optimal draft length, as formalized in Theorem~\ref{prop:unimodality_lambert}.

\begin{Theorem}[Unimodality and Optimal Draft Length]
\label{prop:unimodality_lambert}
The continuous goodput $\tilde{\tau}(\tilde{L})$ in~\eqref{eq:continuous_obj} is strictly unimodal for $\tilde{L}>0$.
If $\frac{T^{\mathsf{ver}}}{\vartheta^\star}>\frac{1-\alpha}{\alpha |\ln \alpha|}$, then $\tilde{\tau}(\tilde{L})$ admits a unique interior maximizer, and the optimal integer draft length $L^\star$ is obtained by
\begin{equation}
\label{eq:integer_closed_form}
    L^\star = \arg\max_{L \in \{\lfloor \tilde{L}^\star \rfloor,\, \lceil \tilde{L}^\star \rceil\}} \tau(L),
\end{equation}
where the continuous optimum $\tilde{L}^\star$ is given in closed form by
\begin{equation}
    \tilde{L}^\star = -\frac{\ln\!\left( -W_{-1}\!\left( -\alpha^{\frac{T^{\sf{ver}}}{\vartheta^\star}-1}\!/e \right) \right)}{\ln \alpha} - 1,
    \label{eq:closed_form_solution}
\end{equation}
with $W_{-1}(\cdot)$ denoting the lower branch of the Lambert $W$ function and $\vartheta^\star$ obtained from~\eqref{eq:vartheta_star_root}.
Otherwise, $\tilde{\tau}(\tilde{L})$ decreases monotonically in $\tilde{L}>0$, and the optimal draft length reduces to $L^\star=1$.
\end{Theorem}
\begin{proof}
    (See Appendix~\ref{proof:Lstar_closed})
\end{proof}

To gain more insights into the optimal solution, we analyze the monotonicity of the optimal draft length $L^\star$ with respect to key system parameters: the verification overhead $T^{\mathsf{ver}}$, the per-token multi-access latency $\vartheta^\star$, and the acceptance rate $\alpha$.
The derivation is provided in Appendix~\ref{app:Lstar_monotonicity}. Fig.~\ref{fig:lstar_param} illustrates how the optimal draft length varies with the key system parameters and motivates the design insights summarized in Remark~\ref{rem:insights_from_closed_form}.

\begin{figure*}[t]
    \centering

    \subfigure[Effect of verification overhead $T^{\mathsf{ver}}$.]{
        \includegraphics[width=0.3\textwidth]{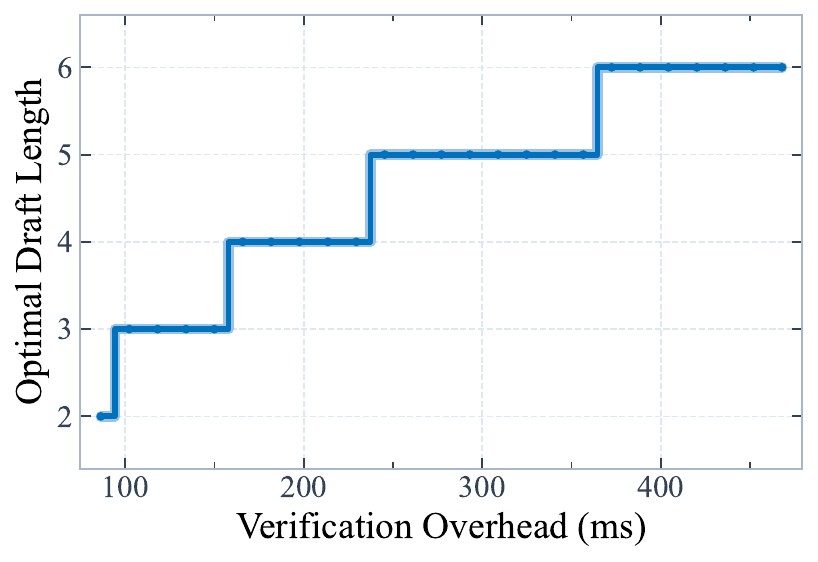}
        \label{fig:lstar_tval}
    }
    \hfill
    \subfigure[Effect of per-token multi-access latency $\vartheta^\star$.]{
        \includegraphics[width=0.3\textwidth]{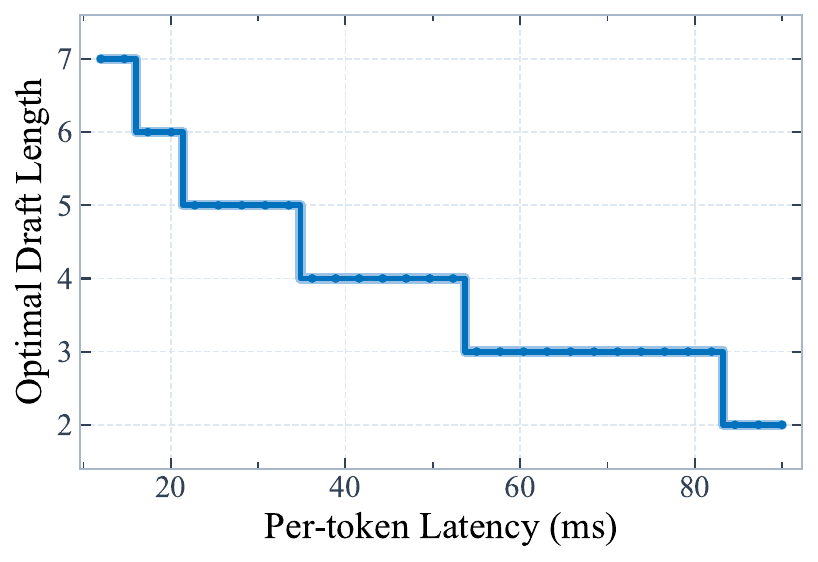}
        \label{fig:lstar_vartheta}
    }
    \hfill
    \subfigure[Effect of acceptance rate $\alpha$.]{
        \includegraphics[width=0.3\textwidth]{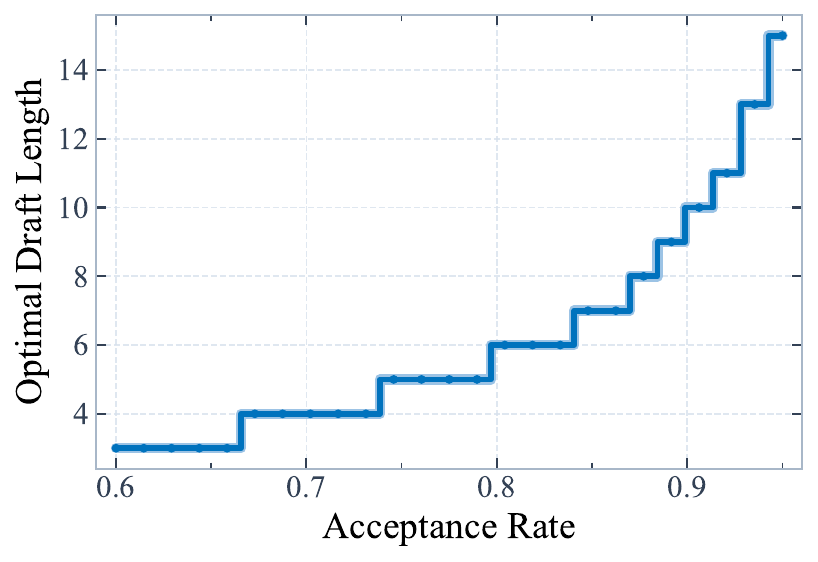}
        \label{fig:lstar_alpha}
    }
    \caption{Optimal uniform draft length under varying system parameters for the Llama-2 Pair. In each subfigure, one parameter is varied while the other parameters are fixed.}
    \label{fig:lstar_param}
\end{figure*}

\begin{Remark}[Effects of System Parameters]
\label{rem:insights_from_closed_form}
The optimal draft length increases with the verification overhead $T^{\mathsf{ver}}$ and decreases with the per-token multi-access latency $\vartheta^\star$, reflecting the need to amortize a costly verification step over more draft tokens when verifying is expensive relative to drafting.
Moreover, the optimal draft length increases with $\alpha$, and as $\alpha\to 1$, $\tilde{L}^\star$ grows as $\mathcal{O}\!\left(1/\sqrt{-\ln\alpha}\right)$, dominating the dependence on $T^{\mathsf{ver}}$ and $\vartheta^\star$. 
This emphasizes the importance of improving the acceptance rate via better SLM--LLM alignment in Multi-SPIN.
\end{Remark}

\section{Multi-access Draft Control with Heterogeneous Lengths}
\label{sec:joint_scheduling}

In this section, we relax the assumption of uniform draft length made in the preceding section and explore the case permitting heterogeneous draft lengths across users. This introduces a new dimension for goodput enhancement but complicates the multi-access draft-length control. To design the optimal control algorithm, the difficulty is overcome by adopting an alternative optimal decomposition approach, which reduces the high-dimensional joint optimization to a low-complexity two-dimensional search.

\subsection{Problem Formulation and Decomposition}
\label{subsec:sec5_formulation}

This subsection formulates the multi-access draft control problem with heterogeneous draft lengths and explains why the previous decoupling strategy is no longer applicable.

Under heterogeneous draft lengths $\mathcal{L}=\{L_k\}_{k=1}^K$ and device-specific acceptance rates $\{\alpha_k\}$, the expected number of accepted tokens varies across devices.
The aggregate expected number of accepted tokens over one round is given by
\begin{equation}\label{eq:ENk_het}
  \sum_{k=1}^{K}\mathbb{E}[N_k\mid L_k]
  = \sum_{k=1}^{K}\frac{1-\alpha_k^{L_k+1}}{1-\alpha_k}.
\end{equation}
Since the draft length differs across devices, the draft lengths cannot be extracted from the multi-access latency. 
Accordingly, the multi-access latency in~\eqref{eq:Tde} generalizes to
\begin{equation}\label{eq:Tde_het}
  T^{\mathsf{ma}}(\mathcal{B},\mathcal{L})
  = \max_{k}\; \left\{L_k\!\left(T_k^{\mathsf{S}}+\frac{Q_{\mathsf{tok}}}{B_k\,r_k}\right)\right\}.
\end{equation}
The corresponding E2E latency of one Multi-SPIN round is then given by $T^{\sf{e2e}}(\mathcal{B},\mathcal{L})=T^{\mathsf{ma}}(\mathcal{B},\mathcal{L})+T^{\sf{ver}}$.
Substituting the above expressions into~\eqref{eq:e2e_metric} yields the sum goodput in the heterogeneous draft-length setting, given by
\begin{equation}
    \label{eq:general_goodput}
    \tau(\mathcal{B},\mathcal{L}) = \frac{ \sum_{k} \mathbb{E}[N_k\mid L_k] }{T^{\sf ma}(\mathcal{B},\mathcal{L})+T^{\sf{ver}}}.
\end{equation}
Extending the multi-access draft control problem to heterogeneous draft lengths yields its general form, formulated as
\begin{align*}
  (\mathrm{P2})\quad
  \max_{\mathcal{B},\,\mathcal{L}} &\quad \tau(\mathcal{B},\mathcal{L}) \notag\\
  \mathrm{s.t.} &\quad L_k\in\mathbb{Z}_{+},\;\;\forall\,k, \notag\\
  &\quad B_k>0,\;\;\forall\,k, \notag\\
  &\quad \textstyle\sum_{k}B_k\le B.
\end{align*}

Problem~(P2) is a \emph{mixed-integer nonlinear program} (MINLP) for which a brute-force solution approach is NP-hard.
The decomposition method for problem~(P1) is no longer optimal since the heterogeneous draft lengths in~\eqref{eq:Tde_het} couple draft lengths and bandwidths inside the $\max$ operator, such that the bandwidth allocation can no longer be optimized independently. 
Nevertheless, an alternative method for optimal decomposition is possible as shown in Lemma~\ref{lem:seq_opt}. 
It is based on the observation that, in the objective in~\eqref{eq:general_goodput}, only the multi-access latency term, $T^{\sf{ma}}(\mathcal{B}, \mathcal{L})$, depends on the bandwidth allocation, $\mathcal{B}$.

\begin{Lemma}[Optimality of Decomposition]
\label{lem:seq_opt}
Since the constraints of problem~(P2) on $\mathcal{B}$ and $\mathcal{L}$ are decoupled, the joint optimization decomposes without loss of optimality into the outer problem over the draft lengths,
\begin{align*}
  (\mathrm{P2.1})\quad
  \max_{\mathcal{L}} &\quad \tau\!\bigl(\mathcal{B}^\star(\mathcal{L}),\,\mathcal{L}\bigr) \notag\\
  \mathrm{s.t.} &\quad L_k\in\mathbb{Z}_{+},\;\;\forall\,k,
\end{align*}
and the inner problem over the bandwidth allocation,
\begin{align*}
  (\mathrm{P2.2})\quad
  \min_{\mathcal{B}} &\quad T^{\mathsf{ma}}(\mathcal{B},\mathcal{L}) \notag\\
  \mathrm{s.t.} &\quad B_k>0,\;\;\forall\,k, \notag\\
  &\quad \textstyle\sum_{k}B_k\le B,
\end{align*}
where $\mathcal{B}^\star(\mathcal{L})$ denotes the optimal bandwidth allocation of the inner problem~(P2.2).
\end{Lemma}

The two sub-problems are solved in turn: Section~\ref{subsec:sec5_deba} solves the inner problem~(P2.2) to obtain $\mathcal{B}^\star(\mathcal{L})$, and Section~\ref{subsec:sec5_optL_het} substitutes it into the outer problem~(P2.1) to determine the optimal draft lengths.

\subsection{Optimal Bandwidth Allocation}
\label{subsec:sec5_deba}

This subsection characterizes the structure of optimal bandwidth allocation for given draft lengths and uses it to eliminate the dependence on the multi-dimensional variable $\mathcal{B}$ in problem~(P2).
For a given set of draft lengths, the aggregate expected number of accepted tokens in~\eqref{eq:ENk_het} is fixed, so maximizing the goodput~$\tau(\mathcal{B},\mathcal{L})$ over~$\mathcal{B}$ reduces to minimizing the multi-access latency~$T^{\mathsf{ma}}(\mathcal{B},\mathcal{L})$ in~\eqref{eq:Tde_het}.

Solving the inner problem~(P2.2), the optimal bandwidth allocation exhibits a \emph{latency-equalization} structure: at optimality, all devices have an equalized multi-access latency, i.e., the minimum latency of problem~(P2.2).
Consequently, the optimal bandwidth allocated to each device is fully determined by this equalized latency and the draft lengths, as provided in Lemma~\ref{lem:optimal_bandwidth}.

\begin{Lemma}
[Structure of Optimal Bandwidth Allocation]
\label{lem:optimal_bandwidth}
For any set of draft lengths $\mathcal{L}$, the optimal bandwidth allocation that solves problem~(P2.2) is given by
\begin{equation}
B_k(\mathcal{L})
=
\frac{Q_{\mathsf{tok}}L_k}{r_k\!\left(\varphi-L_k T_k^{\mathsf S}\right)},
\qquad \forall k,
\label{eq:sec5_Bk_star}
\end{equation}
where $\varphi$ is the equalized multi-access latency, determined as the unique root of
\begin{equation}
\sum_{k=1}^{K}\frac{Q_{\mathsf{tok}}L_k}{r_k\!\left(\varphi-L_k T_k^{\mathsf S}\right)}=B,
\qquad \varphi>\max_{k} L_k T_k^{\mathsf S}.
\label{eq:sec5_varphi_root}
\end{equation}
\end{Lemma}

Lemma~\ref{lem:optimal_bandwidth} provides two structural insights for the subsequent draft-length control.
First, both \(\varphi\) and \(B_k(\mathcal{L})\) are monotonically increasing with respect to each individual draft length \(L_k\). This follows by implicitly differentiating the equalized-latency relation~\eqref{eq:sec5_varphi_root}, which gives $\partial\varphi^\star/\partial L_j>0$, and substituting the result into~\eqref{eq:sec5_Bk_star} to obtain $\partial B_j^\star/\partial L_j>0$ for $K\ge2$.
Hence, assigning a longer draft length to any device increases the equalized latency \(\varphi\) and requires more bandwidth for that device to prevent it from becoming the straggler. 
Second, this lemma reduces problem~(P2) from the $2K$ variables $\mathcal{B}$ and $\mathcal{L}$ to only $K+1$, the equalized latency $\varphi$ and the draft lengths $\mathcal{L}$, laying the foundation for the further reduction in the next subsection.

\subsection{Optimal Draft Control}
\label{subsec:sec5_optL_het}
Building on Lemma~\ref{lem:optimal_bandwidth}, this subsection solves for the structure of optimal draft lengths under the reduced problem formulation.
Substituting the bandwidth--draft-length relationship in~\eqref{eq:sec5_Bk_star} into the goodput eliminates the explicit dependence on~$\mathcal{B}$.
Since Lemma~\ref{lem:optimal_bandwidth} equalizes all multi-access latencies to~$\varphi$, the E2E latency reduces to $\varphi + T^{\mathsf{ver}}$, and the goodput becomes
\begin{equation}
\label{eq:throughput_length}
    \tau(\mathcal{L})=\frac{\sum_k \mathbb{E}[N_k\mid L_k]}{\varphi + T^{\mathsf{ver}}},
\end{equation}
where $\sum_k \mathbb{E}[N_k\mid L_k]$ is the aggregate expected output in~\eqref{eq:ENk_het}, and the equalized latency $\varphi$ is implicitly determined by $\mathcal{L}$ through the bandwidth constraint~\eqref{eq:sec5_varphi_root}.

Consequently, the sum goodput $\tau$ depends only on $\mathcal{L}$, and the outer problem~(P2.1) reduces to~(P2.1a) below,
\begin{subequations}\label{eq:sec5_P5}
\begin{align}
(\mathrm{P2.1 a}) \quad
\max_{\mathcal{L}} \;\; & \tau\left(\mathcal{L}\right) \notag \\
\text{s.t.} \;\; & L_k \in \mathbb{Z}_{+}, \;\; \forall k,
\label{eq:sec5_P5_c1} \\
& \sum_{k=1}^{K}\frac{Q_{\mathsf{tok}}L_k}{r_k(\varphi-L_k T_k^{\mathsf S})} = B,
\label{eq:sec5_P5_c2} \\
& 0 < L_k < \frac{\varphi}{T_k^{\mathsf S}}, \quad \forall k,
\label{eq:sec5_P5_c3}
\end{align}
\end{subequations}
where~\eqref{eq:sec5_P5_c2} encodes the equalized-latency condition from Lemma~\ref{lem:optimal_bandwidth}, and~\eqref{eq:sec5_P5_c3} ensures positive draft lengths and bandwidths.

We next analyze the problem~(P2.1a) and relax it to enable a tractable solution.
Problem~(P2.1a) is a MINLP whose intractability stems from two sources: the integer constraint on $L_k$ and the nonlinear equality constraint~\eqref{eq:sec5_P5_c2}, which implicitly couples all draft lengths through the shared variable~$\varphi$.
To obtain a tractable reformulation, we apply two relaxations: (i)~relaxing the integer draft lengths to continuous variables $\tilde{L}_k >0$, and (ii)~relaxing the equality in~\eqref{eq:sec5_P5_c2} to the inequality
\begin{equation}
\label{eq:relax_constraint}
    \sum_{k=1}^{K}\frac{Q_{\mathsf{tok}}\tilde{L}_k}{r_k(\varphi-\tilde{L}_k T_k^{\mathsf S})} \le B.
\end{equation}
The relaxations do not alter the optimal solution, as proved in Appendix~\ref{app:pf_tightness_bandwidth}.
The resulting relaxed problem is given by
\begin{equation*}
\begin{aligned}
(\mathrm{P2.1b}) \quad
\max_{\tilde{\mathcal{L}}} \;\; & \tau(\tilde{\mathcal{L}}) \\
\text{s.t.} \;\; & \tilde{L}_k > 0, \;\; \forall k,\\
& \eqref{eq:sec5_P5_c3}\; \& \;\eqref{eq:relax_constraint}.
\end{aligned}
\label{eq:sec5_P5_relax}
\end{equation*}

We now derive the structure of optimal draft lengths under the relaxed problem.
For any fixed $\varphi$, problem~(P2.1b) maximizes a strictly concave objective over a convex feasible set, so the \emph{Karush--Kuhn--Tucker}~(KKT) conditions are necessary and sufficient for global optimality~\cite{boyd2004convex}.
Solving these conditions yields a closed-form structure of the optimal draft lengths, as stated in Proposition~\ref{prop:optimal_lengths}.

\begin{Proposition}[Optimal Draft Length]
\label{prop:optimal_lengths}
For a given equalized latency $\varphi$ and KKT multiplier $\lambda > 0$, the optimal integer draft length for device $k$ is obtained by
\begin{equation}
    L_k(\varphi,\lambda) = \operatorname{round} \left(\tilde{L}_{k}(\varphi, \lambda)\right),
\label{eq:rounding_length}
\end{equation}
where the continuous optimum is given by
\begin{equation}
\tilde{L}_{k}(\varphi,\lambda)
=
\frac{\varphi}{T_k^{\mathsf S}}
+\frac{2}{\ln \alpha_k}\,
W_0\!\left(
\frac{\alpha_k^{-\frac{\varphi}{2 T_k^{\mathsf S}}}}{2 T_k^{\mathsf S}}
\sqrt{\frac{\lambda\, Q_{\mathsf{tok}} \,\varphi\, |\ln \alpha_k|}{r_k\, \alpha_k\,(1-\alpha_k)^{-1}}}
\right).
\label{eq:sec5_Lk_closed}
\end{equation}
Here, $W_0(\cdot)$ denotes the principal branch of the Lambert $W$ function.
The two scalar variables $\varphi$ and $\lambda$ are uniquely determined by the KKT stationarity condition and the active bandwidth constraint in~\eqref{eq:sec5_varphi_root}.
\end{Proposition}
\begin{proof}
(See Appendix~\ref{app:pf_prop_optimal_lengths}.)
\end{proof}

Proposition~\ref{prop:optimal_lengths} characterizes the heterogeneous draft lengths through the device-specific parameters $(\alpha_k, T_k^{\mathsf S}, r_k)$ and the two shared scalars $(\varphi, \lambda)$.
More importantly, when combined with the structure of optimal bandwidth in Lemma~\ref{lem:optimal_bandwidth}, it reveals a qualitative shift in the bandwidth allocation principle, as summarized in Remark~\ref{rem:insights_independent}.

\begin{Remark}[Bandwidth Allocation under Heterogeneous Draft Lengths]
\label{rem:insights_independent}
As proved in Appendix~\ref{app:pf_alpha_bandwidth_ordering}, the optimal bandwidth $B_k^\star$ is strictly increasing in $\alpha_k$.
In contrast to the uniform-length regime, where bandwidth compensates devices with weaker C$^2$ capabilities (See Lemma~\ref{prop:teba_solution}), the optimal bandwidth under the heterogeneous-length regime favors devices with higher acceptance rates.
\end{Remark}

\begin{algorithm}[tb]
\caption{Multi-access Draft Control Algorithm}
\label{alg:joint_scheduling}
\begin{algorithmic}[1]
\REQUIRE System parameters $B$, $Q_{\mathsf{tok}}$, $T^{\mathsf{ver}}$; device parameters $\{T_k^{\mathsf S}, r_k, \alpha_k\}_{k=1}^K$; bounded search grids $\Phi\subseteq[\underline{\varphi},\overline{\varphi}]$ and $\Lambda\subseteq[\underline{\lambda},\overline{\lambda}]$.
\STATE Initialize $\tau^\star\leftarrow 0$.
\FOR{each $(\varphi,\lambda)\in \Phi\times\Lambda$}
    \STATE Compute draft lengths $\mathcal{L}=\{L_k\}_{k=1}^K$ via~\eqref{eq:sec5_Lk_closed} and~\eqref{eq:rounding_length}.
    \STATE Find the root $\hat{\varphi}$ of~\eqref{eq:sec5_varphi_root}; if none exists, \textbf{continue}.
    \STATE Compute $\tau\leftarrow\tau(\mathcal{L})$ via~\eqref{eq:throughput_length}.
    \IF{$\tau>\tau^\star$}
        \STATE Update $(\tau^\star,\varphi^\star,\mathcal{L}^\star)\leftarrow(\tau,\hat{\varphi},\mathcal{L})$.
    \ENDIF
\ENDFOR
\STATE Compute $\mathcal{B}^\star$ via~\eqref{eq:sec5_Bk_star} using $(\varphi^\star,\mathcal{L}^\star)$.
\ENSURE $\mathcal{L}^\star$ and $\mathcal{B}^\star$.
\end{algorithmic}
\end{algorithm}

\subsection{Joint Optimization Algorithm}
\label{subsec:sec5_algorithm}
We now assemble the closed-form sub-problem solutions into a joint algorithm for problem~(P2), summarized in Algorithm~\ref{alg:joint_scheduling}.
In the input, we initialize the system parameters, device-specific parameters, and the search grids $\Phi$ and $\Lambda$ specified in Appendix~\ref{app:search_ranges}.
For each candidate pair $(\varphi,\lambda)$, the algorithm computes the draft lengths using Proposition~\ref{prop:optimal_lengths} (Step~3), checks feasibility by evaluating the bandwidth constraint~\eqref{eq:sec5_varphi_root} under the equalized-latency structure (Step~4), and updates the current best solution whenever a higher goodput is attained (Steps~5--8).
The optimal heterogeneous draft lengths and bandwidth allocations for problem~(P2) are accordingly given by
\begin{equation}
\label{eq:final_Lk_star}
    L_k^\star = \operatorname{round}\!\left(\tilde{L}_k(\varphi^\star,\lambda^\star)\right), \quad \forall\, k,
\end{equation}
where $\tilde{L}_k(\cdot)$ is defined in~\eqref{eq:sec5_Lk_closed}, and
\begin{equation}
\label{eq:final_Bk_star}
    B_k^\star = \frac{Q_{\mathsf{tok}}\,L_k^\star}{r_k\!\left(\varphi^\star - L_k^\star\, T_k^{\mathsf{S}}\right)}, \quad \forall\, k,
\end{equation}
which is obtained from Lemma~\ref{lem:optimal_bandwidth} (Step~10).

Solving problem~(P2) directly over its $2K$ variables incurs a complexity of $\mathcal{O}(L_{\max}^{K}K)$, where $L_{\max}$ denotes the maximum admissible draft length per device.
Nevertheless, the decomposition removes this exponential dependence.
Specifically, combining Lemma~\ref{lem:optimal_bandwidth} and Proposition~\ref{prop:optimal_lengths} reduces the $2K$ variables to the two scalars $(\varphi,\lambda)$, first from $2K$ to $K+1$ by Lemma~\ref{lem:optimal_bandwidth} and then to $(\varphi,\lambda)$ by Proposition~\ref{prop:optimal_lengths}, with each feasible $(\varphi,\lambda)$ mapping to a unique solution $(\mathcal{L},\mathcal{B})$.
A two-dimensional grid search over $(\varphi,\lambda)$ therefore recovers the near-optimal solution at a complexity of $\mathcal{O}(|\Phi|\,|\Lambda|\,K)$.
Our algorithm replaces the exponential factor $L_{\max}^{K}$ with the grid size $|\Phi|\,|\Lambda|$, scaling linearly with the number of devices.

\section{Experimental Results}
\label{sec:experiment_results}

\subsection{Experiment Settings}
\label{sec:exp_setting}
\subsubsection{Models and Tasks}
We implement SPIN with two model pairs: (i) TinyLlama-1.1B as the on-device SLM paired with Llama-2-7B as the server-side LLM, and (ii) Qwen3.5-0.8B as the on-device SLM paired with Qwen3.5-27B as the server-side LLM~\cite{touvron2023llama,qwen3.5}.
All models are initialized from pre-trained checkpoints available on Hugging Face.\footnote{The adopted checkpoints are publicly available in the Hugging Face repositories \texttt{TinyLlama/TinyLlama-1.1B}, \texttt{meta-llama/Llama-2-7b}, \texttt{Qwen/Qwen3.5-0.8B}, and \texttt{Qwen/Qwen3.5-27B}.}
The prompts for each device are sampled i.i.d.\ from a mixture of datasets covering diverse task categories to capture heterogeneous workloads across devices.
\begin{itemize}
    \item Task-type 1 (Code Generation): MBPP+~\cite{austin2021program}, 378 sanitized Python problems with augmented test cases.
    \item Task-type 2 (Mathematical Reasoning): GSM8K~\cite{cobbe2021training}, grade-school multi-step math word problems.
    \item Task-type 3 (Multi-turn Dialogue): MT-Bench~\cite{zheng2023judging}, 80 two-turn dialogue questions across diverse categories.
    \item Task-type 4 (Reading Comprehension): SQuAD~\cite{rajpurkar2016squad}, span-extraction questions on Wikipedia passages.
\end{itemize}
For each dataset $d$, an empirical acceptance rate $\alpha^{(d)}$ is estimated by running SPIN on a set of sampled prompts and averaging the realized token-level acceptance probabilities in~\eqref{eq:accept_prob}.
The estimates are summarized in Table~\ref{tab:acceptance_rate_tasks}. 
Device $k$ is then assigned $\alpha_k=\alpha^{(d_k)}$ based on its sampled dataset label $d_k$, resulting in heterogeneous $\{\alpha_k\}$ as defined in~\eqref{eq:alpha_def}.

\begin{table}[t]
\centering
\caption{Empirical acceptance rate (mean $\pm$ std) across prompts for each dataset and SLM--LLM pair, computed by averaging realized acceptance probabilities in~\eqref{eq:accept_prob} over drafted tokens per prompt.}
\label{tab:acceptance_rate_tasks}
\renewcommand{\arraystretch}{1.08}
\setlength{\tabcolsep}{6pt}
\begin{tabular}{lcc}
\hline
\textbf{Dataset} &
\textbf{Llama2 1.1B\,\&\,7B} &
\textbf{Qwen3.5 0.8B\,\&\,27B} \\
\hline
MBPP+ &
$0.8582 \,\pm\, \mbox{\scriptsize 0.2472}$ &
$0.8100 \,\pm\, \mbox{\scriptsize 0.3413}$ \\
GSM8K &
$0.7390 \,\pm\, \mbox{\scriptsize 0.3133}$ &
$0.9340 \,\pm\, \mbox{\scriptsize 0.2089}$ \\
MT-Bench &
$0.7393 \,\pm\, \mbox{\scriptsize 0.3127}$ &
$0.9318 \,\pm\, \mbox{\scriptsize 0.2232}$ \\
SQuAD &
$0.7126 \,\pm\, \mbox{\scriptsize 0.3333}$ &
$0.9650 \,\pm\, \mbox{\scriptsize 0.1538}$ \\
\hline
\end{tabular}
\end{table}

\subsubsection{Computation Settings}
\label{subsubsec:exp_compute}

For mobile devices, the per-token
SLM inference latency $\bar{T}^{\mathsf S}$ is measured on an Apple M4 Pro GPU.
The inference latency of each device is independently drawn from $[0.85, 1.15]\times \bar{T}^{\mathsf S}$, capturing device-side compute heterogeneity.
The edge server runs on an NVIDIA A100 GPU.
We formulate the batched verification latency as a function of batch size $K$ and fit the affine model in~\eqref{eq:batched_validation}.
The fitted curve and the measured points are reported in Fig.~\ref{fig:batched_validation_fit}.

\begin{figure}[t]
    \centering
    \subfigure[Llama2-7B]{
        \includegraphics[width=0.45\columnwidth]{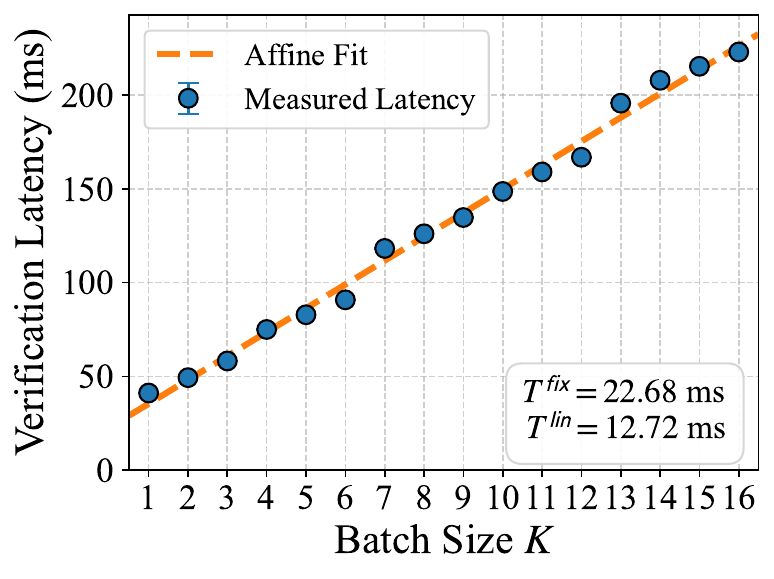}
        \label{fig:llama2_val_fit}
    }
    \hfill
    \subfigure[Qwen3.5-27B]{
        \includegraphics[width=0.45\columnwidth]{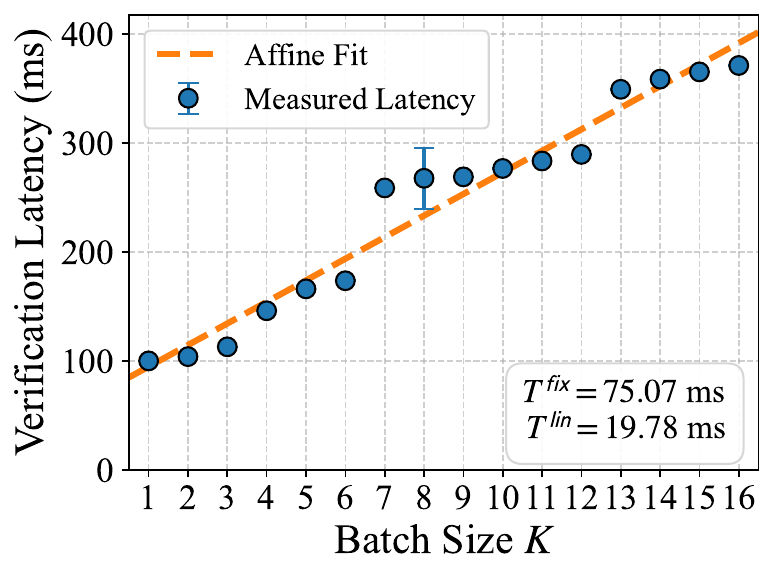}
        \label{fig:qwen_val_fit}
    }
    \caption{Batched verification latency $T^{\sf ver}$ as a function of the batch size $K$. The solid points represent the measured latency on the NVIDIA A100 GPU, and the dashed lines denote the fitted affine models for (a) Llama2-7B and (b) Qwen3.5-27B.}
    \label{fig:batched_validation_fit}
\end{figure}

\subsubsection{Communication Settings}
\label{subsubsec:exp_comm}
Unless otherwise specified, we consider a single-cell OFDMA uplink with $K=20$ devices sharing a total bandwidth of $B=10$\,MHz.
The retained vocabulary size is fixed at $|\hat{\mathcal{V}}|=1024$ for all experiments.
Each device transmits with a constant power spectral density, where the total power budget is $P = 23$\,dBm.
The noise power spectral density is $N_0 = -170$\,dBm/Hz.
The uplink channels follow independent block Rayleigh fading across Multi-SPIN rounds, and the average channel power gains $\{\bar{H}_k\}$ are drawn independently and uniformly at random, yielding average received SNRs in $[18.2,\,22.2]$\,dB.


\begin{figure}[t]
    \centering
    \includegraphics[width=0.75\columnwidth]{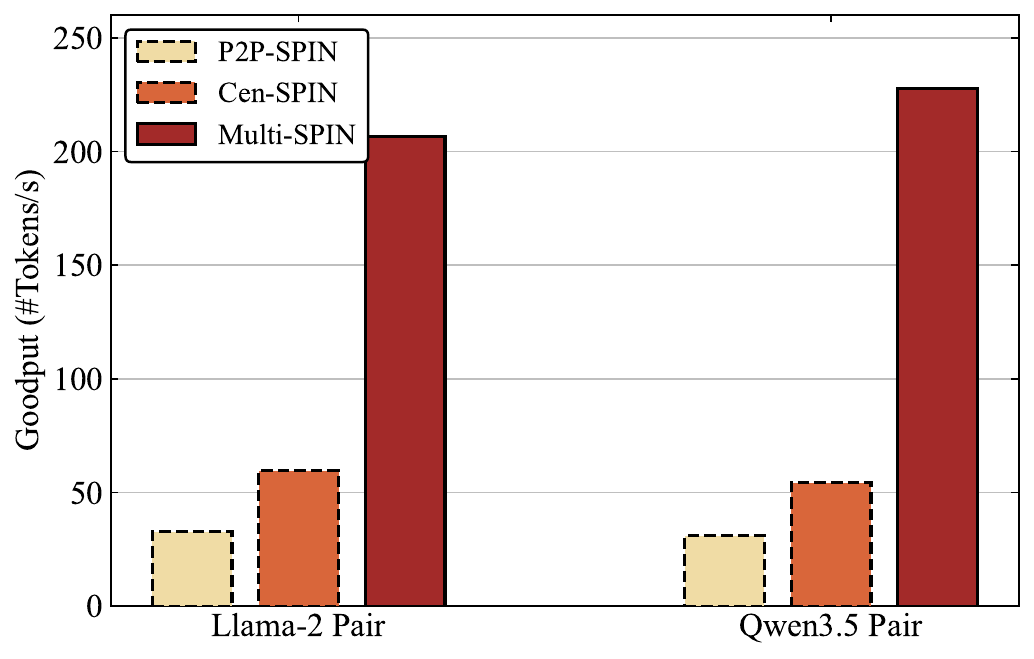}
    \caption{Optimal sum token goodput comparison among P2P-SPIN, Cen-SPIN, and the proposed Multi-SPIN framework for the Llama-2 and Qwen3.5 model pairs.}
    \label{fig:throughput_comparison}
\end{figure}

\subsubsection{Benchmarking Schemes}
We compare Multi-SPIN against two protocol baselines and three optimization baselines, so as to isolate the architectural gain of the Multi-SPIN protocol from the algorithmic gain of its joint bandwidth-and-draft-length optimization. In the optimization comparison, the complete proposed scheme is denoted \emph{Hete-Multi-SPIN} to emphasize that it assigns heterogeneous draft lengths across devices.
For all reported results, the sum goodput is averaged over 50 prompts per device and 10 independent realizations of channel and computational conditions.
In all exhaustive searches for draft-length optimization, the search space is restricted to \(L \in \{1,\ldots,L_{\max}\}\) with \(L_{\max}=25\).

We consider the following protocols as alternatives to the Multi-SPIN architecture.
\begin{itemize}
    \item \emph{P2P-SPIN:} A point-to-point SPIN baseline with a single device, where the device drafts using its local SLM and uploads logits for server verification. The entire bandwidth budget \(B\) is allocated to the single uplink, and the draft length is selected by exhaustive search.
    \item \emph{Cen-SPIN:} A centralized SPIN baseline in which both drafting and verification are executed at the server for each prompt from end devices. The draft length is selected by exhaustive search.
\end{itemize}

The following optimization baselines are constructed within the Multi-SPIN framework by imposing constraints on selected variables.
\begin{itemize}
    \item \emph{Fixed Draft Length and Bandwidth (Fixed BW\&L)}: This baseline fixes the draft length at \(\bar{L}=8\) for all devices and allocates bandwidth uniformly, i.e., \(B_k=B/K\) for all \(k\).
    \item \emph{Multi-SPIN with Uniform Bandwidth Allocation (Uni-BW Multi-SPIN)}: This baseline obtains the per-device heterogeneous draft lengths by solving problem~(P2.1a) with the same relaxation-and-rounding procedure as in Multi-SPIN, while the bandwidth is uniformly allocated, i.e., \(B_k=B/K\) for all \(k\).
    \item \emph{Multi-SPIN with Homogeneous Draft Length (Homo-Multi-SPIN)}: This baseline lets all devices adopt a homogeneous draft length \(L\) selected by exhaustive search to maximize the sum goodput, while the bandwidth is optimized by solving the bandwidth allocation problem~(P1.1).
\end{itemize}

\subsection{Central versus Distributed SPIN Deployment}
Fig.~\ref{fig:throughput_comparison} compares the maximum sum goodput achieved by P2P-SPIN, Cen-SPIN, and Multi-SPIN (i.e., Hete-Multi-SPIN) under the Llama-2 and Qwen3.5 model pairs.
Multi-SPIN consistently delivers the highest goodput in both cases.
For the Llama-2 pair, it reaches approximately $145$~tokens/s, corresponding to about \(2.5\times\) and \(4.6\times\) the goodput of Cen-SPIN and P2P-SPIN, respectively.
A similar trend is observed for the Qwen3.5 pair, where Multi-SPIN attains about $153$~tokens/s and outperforms Cen-SPIN by roughly \(3\times\) despite the substantially higher inference overhead of the larger models.
These results demonstrate the protocol-level gain of Multi-SPIN. 
It combines parallel device-side drafting with batched server-side verification, thereby mitigating the server-side drafting bottleneck in Cen-SPIN while avoiding the inefficient sequential verification pattern of P2P-SPIN.

\begin{figure}[t!]
    \centering
    \subfigure[Llama-2 Pair]{
        \label{fig:baseline_bandwidth_gpt}
        \includegraphics[width=0.45\columnwidth]{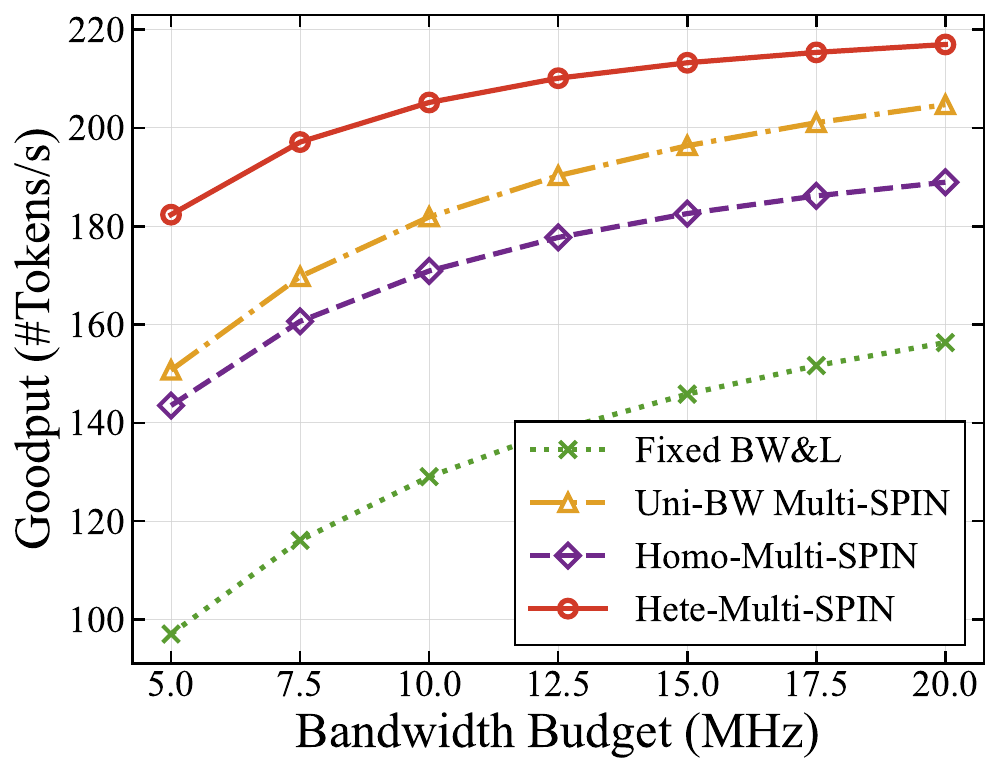}
    }
    \hfill
    \subfigure[Qwen3.5 Pair]{
        \label{fig:baseline_bandwidth_qwen}
        \includegraphics[width=0.45\columnwidth]{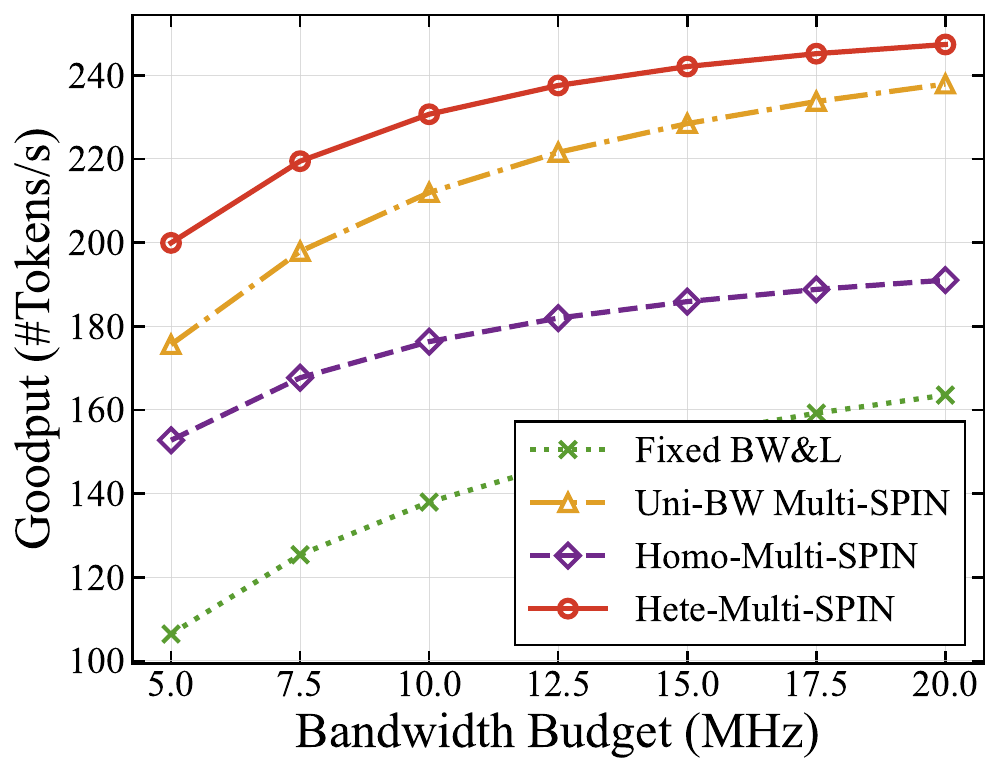}
    }
    \caption{Comparison of sum token goodput across different control schemes under varying bandwidth budgets for both model pairs.}
    \label{fig:baseline_comparison_bandwidth}
\end{figure}

\begin{figure}[t!]
    \centering
    \subfigure[Llama-2 Pair]{
        \label{fig:baseline_user_gpt}
        \includegraphics[width=0.45\columnwidth]{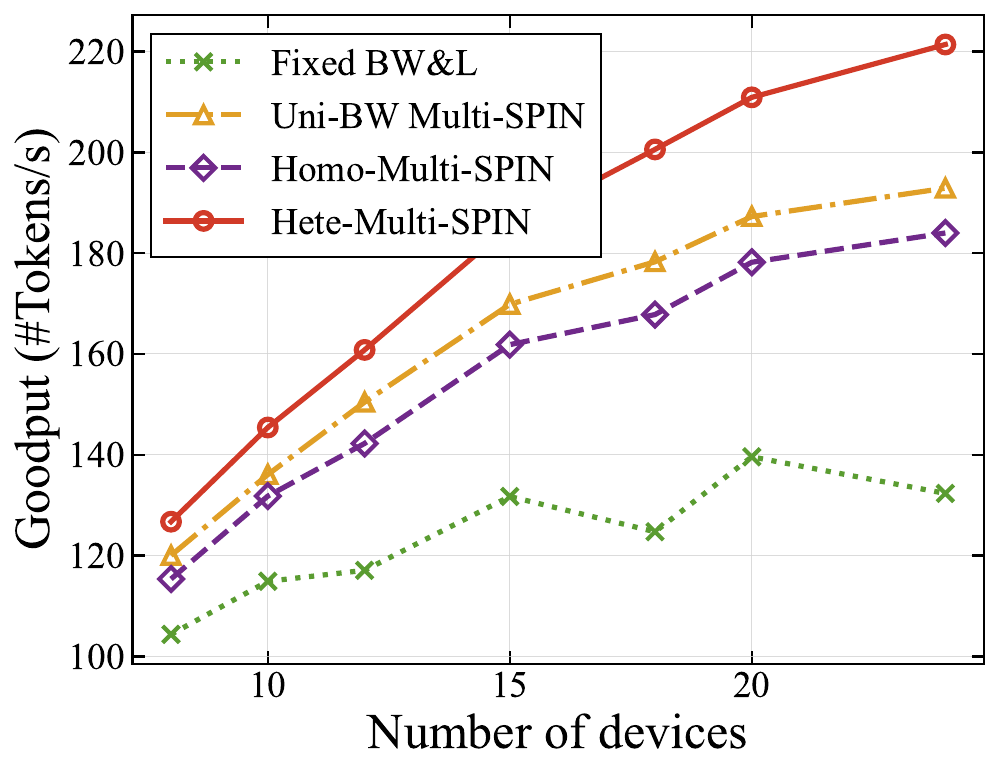}
    }
    \hfill
    \subfigure[Qwen3.5 Pair]{
        \label{fig:baseline_user_qwen}
        \includegraphics[width=0.45\columnwidth]{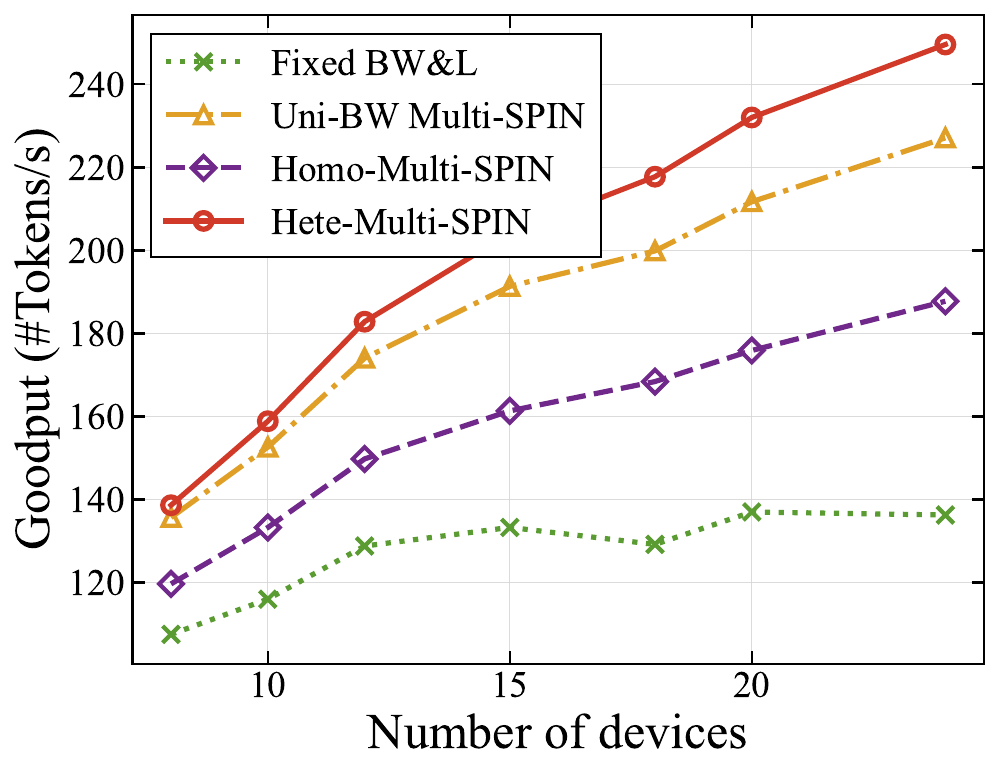}
    }
    \caption{Comparison of sum token goodput across different control schemes under varying numbers of devices for both model pairs.}
    \label{fig:baseline_comparison_user}
\end{figure}

\subsection{Performance of Multi-SPIN Framework}
\label{subsec:exp_framework_perf}

Fig.~\ref{fig:baseline_comparison_bandwidth} compares the sum goodput of Hete-Multi-SPIN and the optimization baselines under varying total bandwidth budgets for the Llama-2 and Qwen3.5 model pairs.
As the budget increases, the goodput of all schemes improves but gradually saturates, marking a transition from a communication-limited regime to a computation-limited one.
Hete-Multi-SPIN attains the highest goodput across both pairs, with its advantage most pronounced in the bandwidth-limited regime, where uniform allocation and fixed draft lengths inflate the multi-access latency under device heterogeneity.
At the smallest bandwidth budget, it improves the goodput over Fixed BW\&L by about \(88\%\) for both model pairs.
This gain narrows as the bandwidth grows and uplink transmission ceases to dominate the E2E latency.
Among the partially optimized baselines, Uni-BW Multi-SPIN consistently outperforms Homo-Multi-SPIN, indicating that adapting draft lengths to device heterogeneity is more beneficial than bandwidth adaptation alone, while jointly optimizing both dimensions yields the largest gain.

Fig.~\ref{fig:baseline_comparison_user} compares the scalability of the schemes as the number of participating devices grows.
Hete-Multi-SPIN scales favorably with more devices, whereas Fixed BW\&L saturates and eventually degrades, producing a widening performance gap.
For the Llama-2 pair, the gain of Hete-Multi-SPIN over Fixed BW\&L grows from about \(21\%\) at small scale to about \(67\%\) at \(K=24\), and for the Qwen3.5 pair from about \(29\%\) to over \(80\%\).
This widening gap arises because uniform resource allocation forces the server to wait for the slowest draft under batching synchronization, an effect that intensifies as device heterogeneity accumulates over a larger pool.
By jointly adapting bandwidth allocation and draft lengths to per-device conditions, Hete-Multi-SPIN suppresses these straggler effects and preserves efficiency at larger system scales.

\section{Concluding Remarks}
\label{sec:conclusions}

This work advocates for the distributed deployment of SPIN to enable cooperative token generation at the network edge. This approach has the advantage of effectively balancing computational loads between edge devices and servers. To materialize this vision, we proposed the Multi-SPIN framework that jointly optimizes multiuser draft lengths and multi-access bandwidth allocation to maximize the sum token goodput. Despite the complexity of these high-dimensional optimization problems, we developed efficient solution strategies via decomposition techniques to successfully derive closed-form solutions for the underlying sub-problems. Thereby, this study establishes draft control as an important mechanism for efficient Multi-SPIN systems by demonstrating its fundamental role in balancing computation loads and communication overhead over devices and servers.

This work opens a new frontier for SPIN-based cooperative token generation at the edge. Among numerous avenues for future investigation, we view the following directions as particularly promising:
\begin{itemize}
    \item \emph{Fairness-Aware Multi-SPIN Control:} Shifting the optimization objective from maximizing the aggregate token goodput to ensuring satisfactory, guaranteed goodput levels for individual devices.
    \item \emph{Adaptive SLM Placement:} Developing mechanisms to dynamically download SLMs from a cloud-based AI library onto edge devices. They can be tailored to users' heterogeneous C$^2$ capabilities and specific task preferences.
    \item \emph{Alternative Multi-Access and Cross-Layer Designs:} Extending the proposed framework beyond OFDMA to incorporate diverse multi-access schemes, such as TDMA, CDMA, and SDMA. This naturally extends to addressing physical-layer challenges (e.g., interference management and adaptive transmission) alongside network-layer issues (e.g., token flow control).
    \item \emph{Advanced Cooperation Strategies:} Investigating more sophisticated cooperation mechanisms, such as task-aware user clustering, to further enhance overarching system efficiency and scalability.
\end{itemize}

\appendix

\subsection{Proof of Theorem~\ref{prop:unimodality_lambert}}
\label{proof:Lstar_closed}

Since the denominator of $\tilde{\tau}'(L)=K\,g(L)/[(1-\alpha)(L\vartheta^\star+T^{\mathsf{ver}})^2]$ is positive, the sign of the derivative of $\tilde{\tau}(L)$ in~\eqref{eq:continuous_obj} is governed by
\[
g(L)=
-\alpha^{L+1}\ln\alpha\,(L\vartheta^\star+T^{\mathsf{ver}})
-\vartheta^\star(1-\alpha^{L+1}).
\]
As $g'(L)=-(\ln\alpha)^2\alpha^{L+1}(L\vartheta^\star+T^{\mathsf{ver}})<0$, $g$ is strictly decreasing, with $g(0)=\alpha |\ln\alpha|T^{\mathsf{ver}}-(1-\alpha)\vartheta^\star$ and $\lim_{L\to\infty} g(L)=-\vartheta^\star<0$. Hence $g$ changes sign at most once and $\tilde{\tau}(L)$ is unimodal, admitting an interior maximizer if and only if $g(0)>0$, i.e., $T^{\mathsf{ver}}/\vartheta^\star>(1-\alpha)/(\alpha|\ln\alpha|)$.
Solving $g(L)=0$ then yields
\[
\tilde{L}^\star=
-\frac{
\ln\!\left(
-W_{-1}\!\left(
-\alpha^{\frac{T^{\mathsf{ver}}}{\vartheta^\star}-1}/e
\right)
\right)
}{\ln\alpha}-1,
\]
and the optimal integer draft length follows by checking the two nearest integers.
Conversely, if $g(0)\le 0$, then $g(L)<0$ for all $L>0$, so $\tilde{\tau}(L)$ is strictly decreasing and its maximum over $L\in\mathbb{Z}_+$ is attained at the boundary $L^\star=1$.

\subsection{Proof of Monotonicity of $\tilde{L}^\star$}
\label{app:Lstar_monotonicity}

Let $t=T^{\mathsf{ver}}/\vartheta^\star$ and $\beta=-\ln\alpha>0$. The stationarity condition in Appendix~\ref{proof:Lstar_closed} is
\[
\beta(\tilde{L}^\star+t)+1=e^{\beta(\tilde{L}^\star+1)}.
\]
Implicit differentiation gives $\partial \tilde{L}^\star/\partial t=1/\bigl(e^{\beta(\tilde{L}^\star+1)}-1\bigr)>0$, hence $\partial \tilde{L}^\star/\partial T^{\mathsf{ver}}>0$ and $\partial \tilde{L}^\star/\partial \vartheta^\star<0$. Differentiating the same condition with respect to $\alpha$ likewise gives $d\tilde{L}^\star/d\alpha>0$, so the optimal draft length increases with the acceptance rate.

Finally, as $\alpha\to 1$, a first-order expansion of the stationarity equation yields
\[
\tilde{L}^\star+1\sim \sqrt{\frac{2(t-1)}{-\ln\alpha}}.
\]

\subsection{Proof of Tightness of Bandwidth Constraint}
\label{app:pf_tightness_bandwidth}

For fixed $L$, the numerator of the relaxed objective is constant, and thus
\[
\tau(L)=\frac{\sum_{k=1}^K \mathbb{E}[N_k\mid L_k]}{\varphi+T^{\mathsf{ver}}}
\]
is strictly decreasing in $\varphi$. Meanwhile, over the feasible region $\varphi>\max_k L_kT_k^{\mathsf S}$, the left-hand side of
\[
\sum_{k=1}^K \frac{Q_{\mathsf{tok}}L_k}{r_k(\varphi-L_kT_k^{\mathsf S})}\le B
\]
is also strictly decreasing in $\varphi$. Therefore, if the constraint were slack at an optimum, one could slightly decrease $\varphi$ while preserving feasibility, which would strictly increase the objective. This contradicts optimality. Hence, the relaxed bandwidth constraint must be active at the optimum.

\subsection{Proof of Proposition~\ref{prop:optimal_lengths}}
\label{app:pf_prop_optimal_lengths}

Consider the relaxed problem~(P2.1b) for a fixed $\varphi$. As proved in Appendix~\ref{app:pf_tightness_bandwidth}, the bandwidth constraint is active at the optimum, so the KKT stationarity condition gives
\[
\frac{-\alpha_k^{\tilde{L}_k+1}\ln\alpha_k}{1-\alpha_k}
=
\lambda\frac{Q_{\mathsf{tok}}\varphi}{r_k(\varphi-\tilde{L}_kT_k^{\mathsf S})^2}.
\]
Let $\beta_k=-\ln\alpha_k$ and $y_k=\varphi-\tilde{L}_kT_k^{\mathsf S}$. Then
\[
y_k^2 e^{\frac{\beta_k}{T_k^{\mathsf S}}y_k}
=
\frac{\lambda Q_{\mathsf{tok}}\varphi(1-\alpha_k)}{r_k\beta_k}
\exp\!\left(\beta_k\left(\frac{\varphi}{T_k^{\mathsf S}}+1\right)\right).
\]
Applying the Lambert $W$ function yields the continuous solution
\[
\tilde{L}_{k}^\star(\varphi, \lambda)
=
\frac{\varphi}{T_k^{\mathsf S}}
+\frac{2}{\ln \alpha_k}\,
W_0\!\left(
\frac{\alpha_k^{-\frac{\varphi}{2 T_k^{\mathsf S}}}}{2 T_k^{\mathsf S}}
\sqrt{\frac{\lambda\, Q_{\mathsf{tok}} \,\varphi\, |\ln \alpha_k|}{r_k\, \alpha_k\,(1-\alpha_k)^{-1}}}
\right).
\]
The integer solution is then obtained by the rounding rule in Proposition~\ref{prop:optimal_lengths}.

\subsection{Proof of Remark~\ref{rem:insights_independent}}
\label{app:pf_alpha_bandwidth_ordering}

It suffices to analyze the continuous optimizer in~\eqref{eq:sec5_Lk_closed}, since the rounding in~\eqref{eq:rounding_length} only affects isolated threshold points.
Let $\beta_k \triangleq -\ln \alpha_k > 0$, $c_k \triangleq \beta_k \varphi/(2T_k^{\mathsf S})$, and
$w_k \triangleq
W_0\!\bigl(
\frac{e^{c_k}}{2T_k^{\mathsf S}}
\sqrt{\lambda Q_{\mathsf{tok}}\varphi\,\beta_k(e^{\beta_k}-1)/r_k}
\bigr)$.
Substituting~\eqref{eq:sec5_Lk_closed} into~\eqref{eq:sec5_Bk_star} and simplifying via $\varphi-\tilde L_kT_k^{\mathsf S}=2T_k^{\mathsf S}w_k/\beta_k$ yields
\[
\widehat B_k
=
\frac{Q_{\mathsf{tok}}}{r_kT_k^{\mathsf S}}
\left(\frac{c_k}{w_k}-1\right).
\]
Since $\tilde L_k>0$ implies $0<w_k<c_k$, it remains to show that $c_k/w_k$ decreases in $\beta_k$.
Differentiating $\ln(c_k/w_k)$ with respect to $\beta_k$ via the identity $W_0'(x)=W_0(x)/[x(1+W_0(x))]$ gives
\[
\frac{\mathrm d}{\mathrm d\beta_k}\ln\!\left(\frac{c_k}{w_k}\right)
=
\frac{1}{\beta_k(1+w_k)}
\left(
w_k-c_k+\frac{1}{2}
-\frac{\beta_k e^{\beta_k}}{2\left(e^{\beta_k}-1\right)}
\right).
\]
Since $w_k<c_k$ and $e^{\beta_k}-1<\beta_k e^{\beta_k}$ for all $\beta_k>0$, the right-hand side is strictly negative, so $\mathrm d\widehat B_k/\mathrm d\beta_k<0$. Because $\alpha_k=e^{-\beta_k}$, $\widehat B_k$ is strictly increasing in $\alpha_k$.

\subsection{Practical Search Ranges for $\varphi$ and $\lambda$}
\label{app:search_ranges}

For Algorithm~\ref{alg:joint_scheduling}, we adopt the following bounded search ranges for $\varphi$ and $\lambda$:
\[
\underline{\varphi}
=
\max_k\left(
T_k^{\mathsf S}+\frac{Q_{\mathsf{tok}}}{B r_k}
\right),
\qquad
\overline{\varphi}
=
\max_k
L_{\max}\!\left(
T_k^{\mathsf S}+\frac{KQ_{\mathsf{tok}}}{B r_k}
\right),
\]
where $L_{\max}$ is the prescribed maximum draft length used in the experiments, and
\[
\underline{\lambda}=\epsilon_\lambda,
\qquad
\overline{\lambda}
=
\max_k
\left[
\frac{r_k(\overline{\varphi}-T_k^{\mathsf S})^2}{Q_{\mathsf{tok}}\overline{\varphi}}
\cdot
\frac{-\ln\alpha_k}{1-\alpha_k}\,
\alpha_k^2
\right],
\]
where $\epsilon_\lambda$ is a small positive constant.

\bibliography{Ref}

@article{6Groadmap,
  author = {Letaief, Khaled B. and Chen, Wei and Shi, Yuanming and Zhang, Jun and Zhang, Ying-Jun Angela},
  title = {The Roadmap to {6G}: {AI} Empowered Wireless Networks},
  journal = {IEEE Commun. Mag.},
  volume = {57},
  number = {8},
  pages = {84--90},
  year = {2019},
  doi = {10.1109/mcom.2019.1900271}
}

@article{GX-CM-2020,
  author = {Zhu, Guangxu and Liu, Dongzhu and Du, Yuqing and You, Changsheng and Zhang, Jun and Huang, Kaibin},
  title = {Toward an intelligent edge: Wireless communication meets machine learning},
  journal = {IEEE Commun. Mag.},
  volume = {58},
  number = {1},
  pages = {19--25},
  year = {2020}
}

@article{liu2023resource,
  author = {Liu, Zhiyan and Lan, Qiao and Huang, Kaibin},
  title = {Resource allocation for multiuser edge inference with batching and early exiting},
  journal = {IEEE J. Sel. Areas Commun.},
  volume = {41},
  number = {4},
  pages = {1186--1200},
  year = {2023},
  doi = {10.1109/jsac.2023.3242724}
}

@article{qu2025partialloading,
  author = {Qu, Guanqiao and Chen, Qian and Chen, Xianhao and Huang, Kaibin and Fang, Yuguang},
  title = {{PartialLoading}: User scheduling and bandwidth allocation for parameter-sharing edge inference},
  journal = {arXiv:2503.22982},
  year = {2025},
  eprint = {2503.22982},
  archivePrefix = {arXiv},
}

@article{shao2025ai,
  author = {Shao, Jiawei and Li, Xuelong},
  title = {{AI} flow at the network edge},
  journal = {IEEE Netw.},
  volume = {40},
  number = {1},
  pages = {330--336},
  year = {2026},
  doi = {10.1109/mnet.2025.3541208}
}

@inproceedings{leviathan2023fast,
  author = {Leviathan, Yaniv and Kalman, Matan and Matias, Yossi},
  title = {Fast inference from transformers via speculative decoding},
  booktitle = {Proc. Int. Conf. Mach. Learn. (ICML)},
  pages = {19274--19286},
  address = {Honolulu, HI, USA},
  year = {2023},
}

@inproceedings{kwon2023efficient,
  author = {Kwon, Woosuk and Li, Zhuohan and Zhuang, Siyuan and Sheng, Ying and Zheng, Lianmin and Yu, Cody Hao and Gonzalez, Joseph and Zhang, Hao and Stoica, Ion},
  title = {Efficient memory management for large language model serving with {PagedAttention}},
  booktitle = {Proc. ACM Symp. Operating Syst. Princ. (SOSP)},
  pages = {611--626},
  address = {Koblenz, Germany},
  year = {2023},
  doi = {10.1145/3600006.3613165}
}

@article{zhang2024edgeshard,
  author = {Zhang, Mingjin and Shen, Xiaoming and Cao, Jiannong and Cui, Zeyang and Jiang, Shan},
  title = {{EdgeShard}: Efficient {LLM} inference via collaborative edge computing},
  journal = {IEEE Internet Things J.},
  volume = {12},
  number = {10},
  pages = {13119--13131},
  year = {2025},
  doi = {10.1109/jiot.2024.3524255}
}

@inproceedings{hao2024hybrid,
  author = {Hao, Zixu and Jiang, Huiqiang and Jiang, Shiqi and Ren, Ju and Cao, Ting},
  title = {Hybrid {SLM} and {LLM} for edge-cloud collaborative inference},
  booktitle = {Proc. Workshop Edge Mobile Found. Models (EdgeFM)},
  pages = {36--41},
  address = {Minato-ku, Tokyo, Japan},
  year = {2024},
  doi = {10.1145/3662006.3662067}
}

@article{xu2024edgellm,
  author = {Xu, Daliang and Yin, Wangsong and Zhang, Hao and Jin, Xin and Zhang, Ying and Wei, Shiyun and Xu, Mengwei and Liu, Xuanzhe},
  title = {{EdgeLLM}: Fast on-device {LLM} inference with speculative decoding},
  journal = {IEEE Trans. Mobile Comput.},
  volume = {24},
  number = {4},
  pages = {3256--3273},
  year = {2025},
  doi = {10.1109/tmc.2024.3513457}
}

@inproceedings{miao2024specinfer,
  author = {Miao, Xupeng and Oliaro, Gabriele and Zhang, Zhihao and Cheng, Xinhao and Wang, Zeyu and Zhang, Zhengxin and Wong, Rae Ying Yee and Zhu, Alan and Yang, Lijie and Shi, Xiaoxiang and others},
  title = {{SpecInfer}: Accelerating large language model serving with tree-based speculative inference and verification},
  booktitle = {Proc. ACM Int. Conf. Archit. Support Program. Lang. Oper. Syst. (ASPLOS)},
  pages = {932--949},
  address = {La Jolla, CA, USA},
  year = {2024},
  doi = {10.1145/3620666.3651335}
}

@article{chen2023accelerating,
  author = {Chen, Charlie and Borgeaud, Sebastian and Irving, Geoffrey and Lespiau, Jean-Baptiste and Sifre, Laurent and Jumper, John},
  title = {Accelerating large language model decoding with speculative sampling},
  journal = {arXiv:2302.01318},
  year = {2023},
  eprint = {2302.01318},
  archivePrefix = {arXiv}
}

@article{li2019edge,
  author = {Li, En and Zeng, Liekang and Zhou, Zhi and Chen, Xu},
  title = {Edge {AI}: On-demand accelerating deep neural network inference via edge computing},
  journal = {IEEE Trans. Wireless Commun.},
  volume = {19},
  number = {1},
  pages = {447--457},
  year = {2020},
  doi = {10.1109/twc.2019.2946140}
}

@article{shao2021communication,
  author = {Shao, Jiawei and Zhang, Jun},
  title = {Communication-computation trade-off in resource-constrained edge inference},
  journal = {IEEE Commun. Mag.},
  volume = {58},
  number = {12},
  pages = {20--26},
  year = {2020},
  doi = {10.1109/mcom.001.2000373}
}

@article{yan2022optimal,
  author = {Yan, Jia and Bi, Suzhi and Zhang, Ying-Jun Angela},
  title = {Optimal model placement and online model splitting for device-edge co-inference},
  journal = {IEEE Trans. Wireless Commun.},
  volume = {21},
  number = {10},
  pages = {8354--8367},
  year = {2022},
  doi = {10.1109/twc.2022.3165824}
}

@article{wang2025revisiting,
  author = {Wang, Zhanwei and Zeng, Qunsong and Zheng, Haotian and Huang, Kaibin},
  title = {Revisiting Outage for Edge Inference Systems},
  journal = {arXiv:2504.03686},
  year = {2025},
  eprint = {2504.03686},
  archivePrefix = {arXiv}
}

@article{wang2026airbreath,
  author = {Wang, Zhanwei and Cui, Mingyao and Yang, Huiling and Zeng, Qunsong and Sheng, Min and Huang, Kaibin},
  title = {Airbreath sensing: Protecting over-the-air distributed sensing against interference},
  journal = {IEEE Trans. Wireless Commun.},
  volume = {25},
  pages = {17415--17429},
  year = {2026},
  doi = {10.1109/twc.2026.3694375}
}

@article{zeng2024knowledge,
  author = {Zeng, Qunsong and Wang, Zhanwei and Zhou, You and Wu, Hai and Yang, Lin and Huang, Kaibin},
  title = {Knowledge-based ultra-low-latency semantic communications for robotic edge intelligence},
  journal = {IEEE Trans. Commun.},
  volume = {73},
  number = {7},
  pages = {4925--4940},
  year = {2025},
  doi = {10.1109/tcomm.2024.3511931}
}

@article{zhang2025llm-inference,
  author = {Zhang, Kai and He, Hengtao and Song, Shenghui and Zhang, Jun and Letaief, Khaled B},
  title = {Communication-efficient distributed on-device {LLM} inference over wireless networks},
  journal = {IEEE J. Sel. Topics Signal Process.},
  volume = {19},
  number = {7},
  pages = {1301--1317},
  year = {2025}
}

@article{xue2025wdmoe,
  author = {Xue, Nan and Sun, Yaping and Chen, Zhiyong and Tao, Meixia and Xu, Xiaodong and Qian, Liang and Cui, Shuguang and Zhang, Wenjun and Zhang, Ping},
  title = {{WDMoE}: Wireless Distributed Mixture of Experts for Large Language Models},
  journal = {IEEE Trans. Wireless Commun.},
  volume = {25},
  pages = {559--572},
  year = {2026},
  doi = {10.1109/twc.2025.3585163}
}

@article{wang2026spacemoe,
  author = {Wang, Zhanwei and Yang, Huiling and Sheng, Min and Letaief, Khaled B and Huang, Kaibin},
  title = {{SpaceMoE}: Realizing Distributed Mixture-of-Experts Inference over Space Networks},
  journal = {arXiv:2605.00515},
  year = {2026},
  eprint = {2605.00515},
  archivePrefix = {arXiv}
}

@inproceedings{zheng2025communicationefficientcollaborativellminference,
  author = {Zheng, Ce and Yang, Tingting},
  title = {Communication-efficient collaborative {LLM} inference via distributed speculative decoding},
  booktitle = {Proc. IEEE Int. Conf. Wireless Commun. Signal Process. (WCSP)},
  pages = {1--6},
  address = {Chongqing, China},
  year = {2025},
  doi = {10.1109/wcsp68525.2025.1010651}
}

@inproceedings{oh2024uncertainty,
  author = {Oh, Seungeun and Kim, Jinhyuk and Park, Jihong and Ko, Seung-Woo and Quek, Tony QS and Kim, Seong-Lyun},
  title = {Uncertainty-aware hybrid inference with on-device small and remote large language models},
  booktitle = {Proc. IEEE Int. Conf. Mach. Learn. Commun. Netw. (ICMLCN)},
  pages = {1--7},
  address = {Barcelona, Spain},
  year = {2025},
  doi = {10.1109/icmlcn64995.2025.11140540}
}

@article{zhang2025quantize,
  author = {Zhang, Guangyi and Cai, Yunlong and Yu, Guanding and Popovski, Petar and Simeone, Osvaldo},
  title = {Quantize-Sample-and-Verify: {LLM} acceleration via adaptive edge-cloud speculative decoding},
  journal = {IEEE Commun. Lett.},
  volume = {30},
  pages = {852--856},
  year = {2026},
  doi = {10.1109/lcomm.2026.3651580}
}

@article{zheng2025fastcollaborativeinferencedistributed,
  author = {Zheng, Ce and Zhang, Ke and Chen, Sun and Zhang, Wenqi and Liu, Qiong and Tesfay, Angesom Ataklity},
  title = {Fast collaborative inference via distributed speculative decoding},
  journal = {J. Inf. Intell.},
  volume = {4},
  number = {1},
  pages = {67--85},
  year = {2026},
  doi = {10.1016/j.jiixd.2025.12.008}
}

@article{wen2023task,
  author={Wen, Dingzhu and Liu, Peixi and Zhu, Guangxu and Shi, Yuanming and Xu, Jie and Eldar, Yonina C. and Cui, Shuguang},
  title={Task-Oriented Sensing, Computation, and Communication Integration for Multi-Device Edge {AI}}, 
  journal = {IEEE Trans. Wireless Commun.},
  year={2024},
  volume={23},
  number={3},
  pages={2486-2502},
  doi={10.1109/TWC.2023.3303232}
}

@inproceedings{chen2025spin,
  author = {Chen, Fahao and Li, Peng and Luan, Tom H and Su, Zhou and Deng, Jing},
  title = {{SPIN}: Accelerating large language model inference with heterogeneous speculative models},
  booktitle = {Proc. IEEE Conf. Comput. Commun. (INFOCOM)},
  pages = {1--10},
  address = {London, U.K.},
  year = {2025},
  doi = {10.1109/infocom55648.2025.11044522}
}

@inproceedings{pope2023efficiently,
  author = {Pope, Reiner and Douglas, Sholto and Chowdhery, Aakanksha and Devlin, Jacob and Bradbury, James and Heek, Jonathan and Xiao, Kefan and Agrawal, Shivani and Dean, Jeff},
  title = {Efficiently scaling transformer inference},
  booktitle = {Proc. Mach. Learn. Syst. (MLSys)},
  pages = {606--624},
  address = {Miami Beach, FL, USA},
  year = {2023}
}

@manual{nvidia_a100_datasheet,
  author = {{NVIDIA Corporation}},
  title = {{NVIDIA A100 Tensor Core GPU}: Data Sheet},
  organization = {NVIDIA Corporation},
  year = {2020},
  url = {https://www.nvidia.com/content/dam/en-zz/Solutions/Data-Center/a100/pdf/nvidia-a100-datasheet-nvidia-us-2188504-web.pdf},
  note = {Accessed: Jun. 7, 2026}
}

@inproceedings{qian2024bass,
  author = {Qian, Haifeng and Gonugondla, Sujan Kumar and Ha, Sungsoo and Shang, Mingyue and Gouda, Sanjay Krishna and Nallapati, Ramesh and Sengupta, Sudipta and Ma, Xiaofei and Deoras, Anoop},
  title = {{BASS}: Batched Attention-Optimized Speculative Sampling},
  booktitle = {Proc. Findings Assoc. Comput. Linguistics (ACL)},
  pages = {8214--8224},
  address = {Bangkok, Thailand},
  year = {2024},
  doi = {10.18653/v1/2024.findings-acl.489}
}

@article{su2023synergy,
  author = {Su, Qidong and Giannoula, Christina and Pekhimenko, Gennady},
  title = {The synergy of speculative decoding and batching in serving large language models},
  journal = {arXiv:2310.18813},
  year = {2023}
}

@article{wu2025resource,
  author = {Wu, Hai and Chen, Xu and Huang, Kaibin},
  title = {Resource management for low-latency cooperative fine-tuning of foundation models at the network edge},
  journal = {IEEE Trans. Wireless Commun.},
  volume = {24},
  number = {6},
  pages = {4839--4852},
  year = {2025},
  doi = {10.1109/twc.2025.3544333}
}

@article{you2016energy,
  author = {You, Changsheng and Huang, Kaibin and Chae, Hyukjin and Kim, Byoung-Hoon},
  title = {Energy-efficient resource allocation for mobile-edge computation offloading},
  journal = {IEEE Trans. Wireless Commun.},
  volume = {16},
  number = {3},
  pages = {1397--1411},
  year = {2017},
  doi = {10.1109/twc.2016.2633522}
}

@inproceedings{parallelSD,
  author = {Liu, Tianyu and Li, Yun and Lv, Qitan and Liu, Kai and Zhu, Jianchen and Hu, Winston and Sun, Xiao},
  title = {{PEARL}: Parallel speculative decoding with adaptive draft length},
  booktitle = {Proc. Int. Conf. Learn. Represent. (ICLR)},
  address = {Singapore},
  year = {2025}
}

@article{yang2026optimal,
  author = {Yang, Huiling and Wang, Zhanwei and Huang, Kaibin},
  title = {Optimal batch-size control for low-latency federated learning with device heterogeneity},
  journal = {IEEE Trans. Commun.},
  volume = {74},
  pages = {5232--5247},
  year = {2026},
  doi = {10.1109/tcomm.2026.3666674}
}

@book{boyd2004convex,
  author = {Boyd, Stephen and Vandenberghe, Lieven},
  title = {Convex Optimization},
  address = {Cambridge, U.K.},
  publisher = {Cambridge Univ. Press},
  year = {2004}
}

@article{touvron2023llama,
  author = {Touvron, Hugo and Martin, Louis and Stone, Kevin and Albert, Peter and Almahairi, Amjad and Babaei, Yasmine and Bashlykov, Nikolay and Batra, Soumya and Bhargava, Prajjwal and Bhosale, Shruti and others},
  title = {Llama 2: Open foundation and fine-tuned chat models},
  journal = {arXiv:2307.09288},
  year = {2023},
  eprint = {2307.09288},
  archivePrefix = {arXiv}
}

@misc{qwen3.5,
  author = {{Qwen Team}},
  title = {{Qwen3.5}: Towards Native Multimodal Agents},
  year = {2026},
  month = {Feb.},
  url = {https://qwen.ai/blog?id=qwen3.5},
  howpublished = {Online}
}

@article{austin2021program,
  author = {Austin, Jacob and Odena, Augustus and Nye, Maxwell and Bosma, Maarten and Michalewski, Henryk and Dohan, David and Jiang, Ellen and Cai, Carrie and Terry, Michael and Le, Quoc and others},
  title = {Program synthesis with large language models},
  journal = {arXiv:2108.07732},
  year = {2021},
  eprint = {2108.07732},
  archivePrefix = {arXiv}
}

@article{cobbe2021training,
  author = {Cobbe, Karl and Kosaraju, Vineet and Bavarian, Mohammad and Chen, Mark and Jun, Heewoo and Kaiser, Lukasz and Plappert, Matthias and Tworek, Jerry and Hilton, Jacob and Nakano, Reiichiro and others},
  title = {Training verifiers to solve math word problems},
  journal = {arXiv:2110.14168},
  year = {2021},
  eprint = {2110.14168},
  archivePrefix = {arXiv}
}

@inproceedings{zheng2023judging,
  author = {Zheng, Lianmin and Chiang, Wei-Lin and Sheng, Ying and Zhuang, Siyuan and Wu, Zhanghao and Zhuang, Yonghao and Lin, Zi and Li, Zhuohan and Li, Dacheng and Xing, Eric and others},
  title = {Judging {LLM}-as-a-judge with {MT-Bench} and Chatbot Arena},
  booktitle = {Proc. Adv. Neural Inf. Process. Syst. (NeurIPS)},
  pages = {46595--46623},
  address = {New Orleans, LA, USA},
  year = {2023},
  doi = {10.52202/075280-2020}
}

@inproceedings{rajpurkar2016squad,
  author = {Rajpurkar, Pranav and Zhang, Jian and Lopyrev, Konstantin and Liang, Percy},
  title = {{SQuAD}: 100,000+ questions for machine comprehension of text},
  booktitle = {Proc. Conf. Empirical Methods Natural Lang. Process. (EMNLP)},
  pages = {2383--2392},
  address = {Austin, TX, USA},
  year = {2016},
  doi = {10.18653/v1/d16-1264}
}

@article{lyu2024rethinking,
  title={Rethinking resource management in edge learning: A joint pre-training and fine-tuning design paradigm},
  author={Lyu, Zhonghao and Li, Yuchen and Zhu, Guangxu and Xu, Jie and Poor, H Vincent and Cui, Shuguang},
  journal={IEEE Trans. on Wireless Commun.},
  volume={24},
  number={2},
  pages={1584--1601},
  year={2024},
}
\bibliographystyle{IEEEtran}

\end{document}